\documentclass[journal]{IEEEtran}

\usepackage{circuitikz}
\usepackage{amsmath,amssymb}
\usepackage{graphicx,epstopdf,subfigure}
\usepackage[comma,numbers,square,sort&compress]{natbib}
\usepackage{epstopdf}
\usepackage{epic}
\usepackage{hyperref}

\usepackage{epstopdf}
\usepackage{epsfig}
\usepackage{arydshln}
\usepackage{enumerate}
\usepackage[utf8]{inputenc}
\usepackage[english]{babel}
\usepackage{fancyhdr}
\usepackage{lastpage}
\usepackage{color}
\usepackage{circuitikz}
\usepackage{tikz}
\usepackage{dsfont}
\usepackage{changes}
\usepackage{booktabs}
\usepackage{threeparttable}
\usepackage{animate}
\usepackage{natbib}
\usepackage{multirow}
\usepackage{bigdelim}
\usepackage{changepage}
\usepackage{algorithm,algpseudocode}
\usetikzlibrary{arrows,shapes,chains}
\usetikzlibrary{calc}
\usetikzlibrary{positioning}

\newcommand{\diag}{\mbox{diag}}
\newcommand{\col}{\mbox{col}}

\newtheorem{rem}{Remark}
\newtheorem{Corollary}{Corollary}
\newtheorem{Definition}{Definition}
\newtheorem{prob}{Problem}
\newenvironment{proof}{\noindent{\em Proof:\/}}{\hfill $\Box$\par}
\newtheorem{thm}{Theorem}
\newtheorem{lem}{Lemma}

\newtheorem{ass}{Assumption}

\newcommand{\row}{\textnormal{row}}

\newcommand{\myr}{}

\allowdisplaybreaks
\begin{document}
\onecolumn
{\Large
This paper was originally published in

Wang, S., Pan, Y. J., and Guay, M. (2024). Distributed state estimation for linear time-invariant systems with aperiodic sampled measurement. IEEE Transactions on Control of Network Systems, DOI: 10.1109/TCNS.2024.3355041.

We have spotted a typo in equation (11), adding missing terms "$\arctan(r)$" and "$\operatorname{arctanh}(r)$"
}

 \twocolumn
\title{Distributed State Estimation for Linear Time-invariant Systems with Aperiodic Sampled Measurement}
\author{Shimin~Wang, Ya-Jun~Pan and ~Martin~Guay
\thanks{This work was supported by MITACS and the Natural Sciences and Engineering Research Council (NSERC), Canada. Shimin~Wang and Martin~Guay are with the Department of Chemical Engineering, Queen's University, Kingston,  ON K7L 3N6, Canada (martin.guay@queensu.ca, bellewsm@mit.edu), and Ya-Jun~Pan is with the Department of Mechanical Engineering, Dalhousie University, Halifax, NS, B3H 4R2, Canada (yajun.pan@dal.ca)\\(Corresponding Author: Martin~Guay)}
}
\maketitle
\begin{abstract}
This paper deals with the state estimation of linear time-invariant systems using distributed observers with local sampled-data measurement and aperiodic communication. Each observer agent receives partial information of the system to be observed but does not satisfy the observability condition. Consequently, distributed observers are designed to exponentially estimate the state of the system to be observed by time-varying sampling and asynchronous communication. Additionally, explicit upper bounds on allowable sampling periods for convergent estimation errors are given. Finally, a numerical example is provided to demonstrate the validity of the theoretical results. 
\end{abstract}

\section{Introduction}

For a given linear time-invariant (LTI) system, the distributed state estimation problem intends
to asymptotically estimate the system's state by combining the partial measurements collected from a group of dynamic agents operating over a network \citep*{mitra2018distributed,wu2021design}. The LTI system to be observed takes the following form:
\begin{align}
\dot{x}(t) &= Ax(t),\label{leader}
 \end{align}
where $x(t)\in \mathds{R}^{n}$ is the vector of state variables and $A\in \mathds{R}^{n\times n}$ is the system matrix. Each agent has access to only partial state information of the system in \eqref{leader} and receives local partial measurements of the form:
\begin{align}
 y_i(t)&=C_ix(t),\;i\in \mathcal{V},\label{leaderooutput}
 \end{align}
where $\mathcal{V}$ is the set of all nodes, $y_i(t)\in \mathds{R}^{p_i}$ is the vector of output measurements and $C_i\in \mathds{R}^{p_i\times n}$ is the output matrix for the $i$-th node.

In \cite{olfati2007distributed}, a distributed algorithm using a group of sensors over an undirected graph was designed to estimate the state variables of a LTI system under the assumption that each pair $(A, C_i)$ is observable.
A distributed observer over a general directed graph was proposed in \cite{su2011cooperative} to solve the cooperative output regulation problem. The more general case of communication graphs with switching typologies was considered in \cite{su2012cooperative}. In these results, it is assumed that a subset of the agents have access to the full state vector $x(t)$, such that $C_i\in\{0_{n\times n},I_n\}$. However, the
agents' dynamics can still reconstruct the full state $x(t)$. 
In an attempt to generalize these early results, a distributed estimation scheme was proposed in \cite{park2016design}, in which each agent can estimate the system's state using partial output signals.
Another type of distributed observers was constructed in \cite{mitra2018distributed} for systems that meet a local detectability assumption.  For these systems, it is assumed that the pair $(A,C_{\mathcal{N}_i})$ is observable, where $C_{\mathcal{N}_i}$ contains the output matrix of the $i$-th agent and its neighbours.
The results of \cite{park2016design} and \cite{mitra2018distributed} were generalized to strongly connected graphs in \cite{wang2017distributed}. The approach proposes the design of a reduced-order continuous-time distributed observer that addresses some of the limitations of the results presented in \cite{park2016design}. 
A discrete-time distributed observer design was considered in \cite{wang2019distributed}.

In \cite{wang2017distributed} and \cite{wang2019distributed}, the design of distributed observer was proposed for systems that are jointly observable for which the pair $(A, C)$ is observable with $C=\col(C_1, \cdots, C_N)$.
The jointly observable assumption is the mildest possible restriction as it allows the pair $(A, C_i)$ to be unobservable for each node while enabling the reconstruction of the system's state through the local exchange of information.
In the approach proposed in \cite{wang2017distributed} and \cite{wang2019distributed}, each agent is required to have access to some partial information such that $C_i\neq 0$.
To relax this assumption, a Kalman observable canonical decomposition was used in \cite{han2018simple} to design a full state distributed observer under the jointly observable assumption without requiring $C_i\neq0$.
An improvement of the design of the distributed observer proposed in \cite{han2018simple} and \cite{han2018towards} was developed in  \cite{kim2019completely} by mixing a linear matrix inequality (LMI)-based approach with a reduced-order observer form.
More recently,  a novel design of distributed observers was proposed in \cite{zhang2022decentralized} in which the system \eqref{leader} was transformed to the real Jordan canonical form.
Learning-based approaches were developed in \cite{wang2018adaptive} and \cite{baldi2020distributed} for the design of distributed observers that adaptively estimate the state and parameters of a linear leader system. Distributed observers for systems with nonlinear leader dynamics were presented in \citep*{wu2021design}.
In addition, meaningful and practical considerations of state estimation for a class of linear time-invariant systems with unknown inputs and switching communication topology have been presented in \cite{wang2022robust,yang2022state,cao2023distributed,wang2023distributed} and \cite{yang2023state}, respectively.

It should be noted that all these references are primarily concerned with either continuous-time or discrete-time systems. To date, only limited work has considered the design of estimation techniques for practical aperiodic sampled-data systems. {\myr Distributed state estimation and traditional state estimation for systems with non-uniform sampling were considered in  \cite{li2017robust} and \cite{sferlazza2018time,sferlazza2021state}, respectively. }
For example, the round-robin aperiodic sampled measurements scheme studied in \cite{sferlazza2021state} largely exploits the sequential nature
of the measurement in distributed estimation problem and complements the results presented in \cite{li2017robust}. In particular, a time-varying observer for a linear continuous-time plant with asynchronously sampled measurements was provided in \cite{sferlazza2018time}, which was formulated in the hybrid systems
framework, providing an elegant setting. 

The importance of the communication networks' attributes in the design of distributed observers was demonstrated in a number of studies such as  \cite{su2012cooperative,kim2019completely} and \cite{wu2021design}. 
For example, an analytical relationship between the system matrix $A$ in \eqref{leader} and the minimum real part of the Laplacian matrix's eigenvalues was provided for the continuous-time distributed estimation problem in \cite{su2012cooperative}. 
A sufficient condition was given for both linear and nonlinear system cases under a jointly observable assumption in \cite{kim2019completely} and \cite{wu2021design}, respectively.
An analysis of the impact of the sampling period on consensus behaviour of
second-order systems \citep*{yu2011second,huang2016some} revealed that a consensus cannot be achieved for any sampling period if there exists one eigenvalue of the Laplacian matrix with a nonzero imaginary part.
The interactions between the choice of sampling periods, the network topologies, 
the reference signals, and the related observability of the system were 
fully investigated in \cite{wang2021cooperative}. All non-pathological or pathological
sampling periods were identified.

Motivated by the studies mentioned above, this paper considers distributed observers for linear systems using {local sampling information and computation}. The main contributions are summarized as follows:
\begin{enumerate}
  \item The design of distributed observers with aperiodic sampled-data information is proposed to estimate the state of the to-be-observed system that satisfies a jointly observable assumption.
  \item An estimated bound for the sampling intervals is given to guarantee the convergence of the estimation error. As long as the sampling periods of all agents' dynamics are smaller than this estimated allowable sampling bound, the estimation error will tend to zero exponentially. In addition, we give an algorithm to calculate the explicit upper bound of the sampling periods using a hybrid system technique.
\item Compared with the existing results in \cite{su2011cooperative} and \cite{ding2013network}, the proposed study relaxes the observability requirement of existing results to tackle the jointly observable assumption. Each agent can asymptotically complete the state of the LTI system to-be-observed using only its partial measurements and its neighbors' state estimates.
\end{enumerate}

The rest of this paper is organized as follows. In Section \ref{section2}, the problem is formulated. Some standard assumptions and lemmas are introduced. Section \ref{mainresults} is devoted to the design of distributed observers. A simulation example in Section \ref{numerexam} followed by brief conclusions in Section~\ref{conlu}.

\textbf{Notation:} Let $\|\cdot\|$ denote both the Euclidean norm of a vector and the Euclidean induced matrix norm (spectral norm) of a matrix. $\mathds{R}$ is the set of real numbers. $\mathds{N}$ denotes all natural numbers. $\mathds{Z}$ ($\mathds{Z}_{+}$) is the set of all (positive) integers. $I_n$ denotes the $n\times n$ identity matrix. For $A\in \mathds{R}^{m\times n}$, $\textnormal{Ker}(A)=\{x\in \mathds{R}^n|Ax=0\}$ and $\textnormal{Im}(A) =\{y\in \mathds{R}^m | y=Ax \textnormal{~~for~~some~~} x \in \mathds{R}^n\}$ denote the kernel and range of $A$, respectively. For a subspace $\mathcal{V}\subset \mathds{R}^{n}$, the orthogonal complement of $\mathcal{V}$ is denoted as $\mathcal{V}^{\bot}=\{x\in \mathds{R}^{n}|x^Tv=0, \forall v\in \mathcal{V}\}$. $\otimes$ denotes the Kronecker product of matrices.
$\textbf{0}$ denotes a zero matrix with conformable dimensions. 
For $b_i\in \mathds{R}^{n_i \times p}$, $i=1,\dots,m$,
$\col(b_1,\dots,b_m)\triangleq\big[
                       \begin{smallmatrix}
                          b_1^T & \cdots & b_m^T \\
                        \end{smallmatrix}
                      \big]^T$. For $a_i\in \mathds{R}^{p \times n_i}$, $i=1,\dots,m$,
$\row(a_1,\dots,a_m)\triangleq\big[
                       \begin{smallmatrix}
                          a_1 & \cdots & a_m\\
                        \end{smallmatrix}
                      \big]$. 
For $X_1\in \mathds{R}^{n_1\times m_1},\dots,X_k\in \mathds{R}^{n_k\times m_k}$,  
\begin{align*}\mbox{diag} (X_1,\dots,X_k)\triangleq\left[
                                                      \begin{array}{ccc}
                                                        X_1 &   &   \\
                                                          & \ddots &   \\
                                                          &   & X_k\\
                                                      \end{array}
                                                    \right].
\end{align*}
\section{Problem Formulation and Assumptions}\label{section2}
In this section, we formulate the \emph{Jointly Observable Tracking Problem} for linear multi-agent systems. {\myr To solve this problem, we introduce the design framework depicted schematically in Figure.~\ref{schemetic}.}
\begin{figure*}[htbp]
 \centering
 \includegraphics[trim={5.25cm 18.4cm 5.25cm 4.25cm}, width=0.65\textwidth,clip]{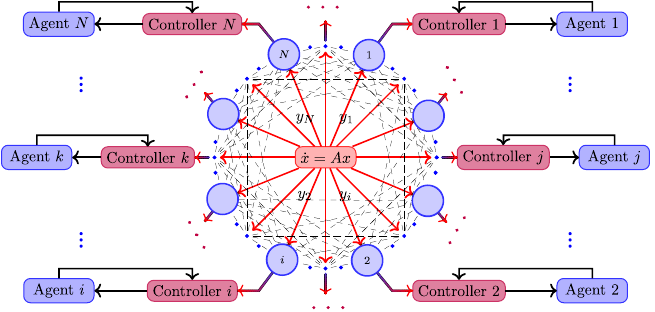}
  \caption{\myr Schematic of Jointly Observable Tracking
Problem}\label{schemetic}
\end{figure*}
\subsection{Agent's dynamics}
We consider the jointly observable network with the system in \eqref{leader} as the LTI system to be observed. The dynamics of each agent take the general form: 
\begin{subequations}\label{follower}
\begin{align}
\dot{\xi}_i(t)=&A_i\xi_i(t)+B_i u_i(t),\;i\in \mathcal{V},\\
y_{\xi_i}(t)=&  F_i\xi_i(t)+D_iu_i(t),
\end{align}
\end{subequations}
where $\xi_i(t)\in \mathds{R}^{n_i}$ is the state variable vector of the $i$-th agent's dynamics, $y_{\xi_i}(t)\in \mathds{R}^{n_{y_{\xi_i}}}$ and $u_i(t)\in \mathds{R}^{m_i}$ is the input used by the agent.

Let $t_0=0$ denote the initial time for the system. We let $\hat{x}_i(t)\in \mathds{R}^n$ denote the local estimate of $x(t)$ for agent $i$ at time moment $t$. Consider a sequence of aperiodic sampling times $t_k$ for $k\in \mathds{N}$.  The discrete-time signal $y_i(t_k)\in \mathds{R}^n$ is the measurement available to agent $i$.   Additionally, each agent samples the state estimates of its neighbours $\hat{x}_j(t_k)$, for $j\in \mathcal{N}_i$, at each sampling instant $t_k$.

\subsection{Graph theory basics}

We introduce some basic elements from graph theory. As in \cite{mitra2018distributed} and \cite{wang2017distributed}, the system composed of \eqref{leader} and \eqref{follower} can be viewed as a multi-agent system with the system to be observed and $N$ agents. The network topology among the multi-agent systems is described by a graph $\mathcal{G}\triangleq\left(\mathcal{V},%
\mathcal{E}\right)$ with $\mathcal{V}\triangleq\{1,\dots,N\}$ and $\mathcal{E}\subseteq\left[\mathcal{V}\right]^2$, which are the
2-element subsets of $\mathcal{V}$. Here, the $i$-th node is associated with the $i$-th agent's dynamics for $i=1,\dots,N$.
For $i=1,\dots,N$, $j=1,\dots,N$, $(j,i) \in {\mathcal{E}} $ if and
only if agent $i$ can receive information from agent $j$. Let $\mathcal{N}_i\triangleq\{j|(j,i)\in \mathcal{E}\}$ denote the neighborhood set of agent $i$. The weighted adjacency matrix of a digraph $\mathcal{G}$ is a nonnegative matrix $\mathcal{A}=[a_{ij}]\in \mathds{R}^{N\times N}$, where $a_{ii}=0$ and $a_{ij}>0\Leftrightarrow (j,i)\in \mathcal{E}$.
Let $\mathcal{{L}}$ be the
Laplacian matrix on graph $\mathcal{G}$, where $l_{ij}$ is the $(i,j)$-th entry of the Laplacian matrix $\mathcal{L}$ with $l_{ii}=\sum_{j=1}^{N}\alpha_{ij}$ and $l_{ij}=-\alpha_{ij}$, $i\neq j$. More details on graph theory can be found in \cite{godsil2013algebraic}.

 In this paper, we consider the design of the input $u(t)$ and a local state observer based on the aperiodic sampled information as follows:
\begin{subequations}\label{pcontrol}\begin{align}
u(t) =&\mathbf{k}_i(\xi_i(t), \hat{x}_i(t)),\\
\dot{\hat{x}}_i(t)=&\textbf{g}_i\Big(\hat{x}_i(t),y_i(t_k), \sum\limits_{j\in{\mathcal{N}}_i}\hat{x}_j(t_k)\Big),\label{generalconrolb}\; t\in [t_k,t_{k+1}),
 \end{align}\end{subequations}
where {\myr $\xi_i$ is directly accessible
for the controller \eqref{pcontrol}}, $\textbf{k}_i(\cdot)$ and $\textbf{g}_i(\cdot)$ are expressions to be designed later, for $i\in \mathcal{V}$. For every $k\in \mathds{N}$, the difference between two adjacent sampling moments is
 $$t_{k+1}-t_k\triangleq h_k,$$
where $h_k\in (0,h_{\max})$ with $h_{\max}$ being some positive number to be determined. In addition, the sampling instants are monotonically increasing sequences satisfying $\lim\limits_{k\rightarrow\infty}t_k=\infty$. As in \cite{laila2002open,karafyllis2007small,nesic2009explicit,sferlazza2018time,oishi2010stability} and \cite{sferlazza2021state}, we define the parameter $h_{\max}$ as an \emph{\myr estimated allowable sampling bound}. While the computation of this quantity is very challenging, its knowledge is imperative to deal with aperiodic sampling.
\subsection{Problem Formulation}
Now we can formulate the \emph{Jointly Observable Tracking Problem} as follows:
\begin{prob}[Jointly Observable Tracking Problem]\label{prob1} Consider the system in \eqref{leader} and \eqref{follower}. Find a distributed control action of the form \eqref{pcontrol} such that for any $\xi_i(0)\in\mathds{R}^{n_i}$,  $\hat{x}_i(0)\in\mathds{R}^{n}$ and $x(0)\in\mathds{R}^{n}$, the closed-loop system satisfies
 $$\myr \lim_{t \to \infty}(y_{\xi_i}(t)-Y_{ix}(t))  =\textbf{0},\;i\in \mathcal{V},$$
{\myr where $x(t)\in \mathds{R}^{n}$ is the state of the to be observed system \eqref{leader} and $Y_{ix}(t)=Y_i x(t)\in \mathds{R}^{n_{y_{\xi_i}}}$ is the tracking signal arising from a matrix $Y_i$ of proper dimensions. }
\end{prob}

It should be noted that, in Problem \ref{prob1}, the $i$-th follower has access to the signal $y_i(t_k)=C_i x_i(t_k)$ only at the discrete-time instant $t_k$ with $k\in \mathds{N}$ and $i\in \mathcal{V}$. Each partial local measurement, $y_i(t)$ is an element of the lumped output $y(t)\triangleq\col(y_1,\cdots,y_N)$ of the  system to be observed in \eqref{leader}, for $i\in \mathcal{V}$.
Compared with the existing work in \cite{su2011cooperative} and \cite{ding2013network}, Problem \ref{prob1} removes the assumption that the full state or observable state of the system to be observed is available to some of the agents.


A key technique to solve the \emph{Jointly Observable Tracking Problem} is the \emph{Sampled-Data Distributed Observer} defined in the following.
%

\begin{Definition}[Sampled-Data Distributed Observer] Given a communication topology $\mathcal{{G}}$, the system  \eqref{generalconrolb} is called a sampled-data distributed observer of the $i$-th node dynamics for the system to be observed \eqref{leader} if there exists globally defined functions $\textbf{g}_i(\cdot)$ and a positive constant $h_{\max}$, such that, for any initial conditions $\hat{x}_i(0)\in\mathds{R}^{n}$ and $x(0)\in\mathds{R}^{n}$, and any sequence $\{t_k,k\in \mathds{N}\}$ satisfying $h_k\in (0,h_{\max})$,
$$\lim_{t\rightarrow\infty}\left(\hat{x}_i(t)-{x}(t)\right)=\textbf{0},~~~~i\in \mathcal{V}.$$
\end{Definition}

\subsection{Assumptions and Lemmas}

We state the following assumptions that will be used in this study.

\begin{ass}\label{ass6} The pair $(A_i, B_i)$ is controllable  $\forall i \in \mathcal{V}$.
\end{ass}

\begin{ass}\label{ass6-ex}  The following linear matrix equations have
solutions $X_i$ and $U_i$ for all $i \in \mathcal{V}$:
\begin{align*}
X_i A =& A_iX_i + B_iU_i,\\
\textbf{0} = &F_iX_i + D_iU_i -Y_i.
\end{align*}
\end{ass}

\begin{ass}\label{ass0} $\mathcal{G}$ is a strongly connected directed graph.
\end{ass}
\begin{ass}\label{ass1} The system in \eqref{leader} is jointly observable in the sense that $(A,C)$ is observable with $C=\textnormal{\col}(C_1,\cdots,C_N)$.
\end{ass}

\begin{rem} {\myr Assumption \ref{ass6-ex} is a standard assumption for the solution of cooperative tracking problems. The linear matrix equations are called regulator equations whose solutions determine the feedforward control gains, as presented in \cite{su2011cooperative}.} For $i\in \mathcal{V}$, we assume that the observability index of $(A, C_i)$ is $v_i$, such that $\textnormal{rank}(\mathcal{O}_i)=v_i$, 
where $\mathcal{O}_i \in \mathds{R}^{(\sum p_i)\times n}$ is the observability matrix and defined {\myr as follows} $\mathcal{O}_i=\textnormal{\col}\left(C_i,C_iA,\cdots,C_iA^{n-1}\right)$. 
For $i\in \mathcal{V}$, the observable subspace and unobservable subspace of $(A,C_i)$ are defined as $\textnormal{Im}(\mathcal{O}_i^T)\subset \mathds{R}^{n}$ and $\textnormal{Ker}(\mathcal{O}_i)\subset \mathds{R}^{n}$, respectively, and satisfy $\textnormal{Ker}(\mathcal{O}_i)^{\bot}=\textnormal{Im}(\mathcal{O}_i^T)$.

For $i\in \mathcal{V}$, let $V_i=\row(V_{ui},V_{oi})\in\mathds{R}^{n\times n}$ be an orthogonal matrix such that $V_iV_i^T=I_n$.
Let $V_{ui}\in\mathds{R}^{n\times (n-v_i)}$ be a matrix such that all columns of $V_{ui}$ are from an orthogonal basis of the $\textnormal{Ker}(\mathcal{O}_i)$ satisfying $\textnormal{Im}(V_{ui})=\textnormal{Ker}(\mathcal{O}_i)$.
Let $V_{oi}\in\mathds{R}^{n\times v_i}$ be a matrix such that all columns of $V_{oi}$ are from an orthogonal basis of the $\textnormal{Im}(\mathcal{O}_i^T)$ satisfying $\textnormal{Im}(V_{oi})=\textnormal{Im}(\mathcal{O}_i^T)$.
%

%
For $i\in \mathcal{V}$, the matrices $A$ and $C_i$ of the system in \eqref{leader} yield the Kalman observability decomposition as follows:
\begin{subequations}\label{decom}\begin{align}
V_i^{T}AV_i
=&\left[
            \begin{array}{cc}
              A_{ui} & A_{ri} \\
              \textbf{0} & A_{oi} \\
            \end{array}
          \right],\\
C_iV_i=&\left[
                              \begin{array}{cc}
                                \textbf{0} & C_{oi} \\
                              \end{array}
                            \right],
\end{align}\end{subequations}
where the pair $(A_{oi},  C_{oi})$ is observable, $A_{oi}\in \mathds{R}^{v_i\times v_i}$, $ A_{ri}\in \mathds{R}^{(n-v_i)\times v_i}$, $A_{ui}\in \mathds{R}^{(n-v_i)\times (n-v_i)}$ and $ C_{oi} \in \mathds{R}^{p_i\times v_i}$ admit the following {\myr matrices}: $A_{ui}=V_{ui}^TAV_{ui}$, $A_{ri}=V_{ui}^TAV_{oi}$, $A_{oi}=V_{oi}^TAV_{oi}$ and $C_{oi}=C_iV_{oi}$. \end{rem}
\begin{rem}\label{diagObser} Let $C_{o}=\textnormal{diag}({C_{o1}},\cdots,{C_{oN}})$, $A_{r}=\textnormal{diag}({A_{r1}},\cdots,{A_{rN}})$, $V_{o}=\textnormal{diag}(V_{o1},\cdots,V_{oN})$, $V_{u}=\textnormal{diag}(V_{u1},\cdots,V_{uN})$, $A_{o}=\textnormal{diag}({A_{o1}},\cdots,{A_{oN}})$ and $A_{u}=\textnormal{diag}({A_{u1}},\cdots,{A_{uN}})$.
\end{rem}

Before proceeding, we review some lemmas proposed in \cite{kim2019completely} and \cite{zhang2015constructing}, which will play important roles in analyzing the convergence of the estimation error.
 \begin{lem}\label{meitopol}\citep*{zhang2015constructing} Suppose that the communication network $\mathcal{G}=(\mathcal{V},\mathcal{E})$ is strongly connected. Let $\theta=\col(\theta_1,\cdots,\theta_N)$  be the left eigenvector of the Laplacian matrix $\mathcal{L}$ associated with the eigenvalue $0$, i.e., $\mathcal{L}^T\theta=\textbf{0}$. Then,
  $\Theta =\textnormal{\diag}(\theta_1,\cdots,\theta_N)>0$ and
$\hat{\mathcal{L}}=\Theta \mathcal{L}+\mathcal{L}^T\Theta\geq0$.
 \end{lem}
\begin{lem}\label{kimPD}\citep*{kim2019completely} Suppose that the communication network $\mathcal{G}=(\mathcal{V},\mathcal{E})$ is strongly connected. Then, the following statements are equivalent:
\begin{enumerate}
  \item  The system in \eqref{leader} is jointly observable;
  \item  The matrix $V_{u}^T\big(\hat{\mathcal{L}}\otimes I_n\big)V_{u}$ is positive definite;
  \item  The matrix $V_{u}^T\left({\mathcal{L}}\otimes I_n\right)V_{u}$ is nonsingular.
\end{enumerate}
\end{lem}

Next, we first introduce some notation related to graphs.

\begin{rem}\label{remarkgraph} Let $\theta_m \triangleq \min\{\theta_1,\cdots,\theta_N\}$ and $\theta_M \triangleq \max\{\theta_1,\cdots,\theta_N\}$. Let $\lambda_{l}$ and $\lambda_{L}$ denote the minimum and maximum eigenvalues of $V_{u}^T(\mathcal{\hat{L}}\otimes I_n)V_{u}$, respectively.
\end{rem}

Before stating the main results of this study, we establish the following lemma. Its proof
can be found in {\myr Appendix \ref{appendixa}}. 
 \begin{lem}\label{lemmaexp} Consider the following sampled-data system
 \begin{align}\label{compactformii}
\dot{z}(t) =&A_{u}z(t)-\gamma V_{u}^T\left(\mathcal{L}\otimes I_n\right)V_{u}z(t_k), \;t\in[t_{k+1}, t_{k}),
\end{align}
where $z(t)=\textnormal{\col}(z_1(t),\cdots,z_N(t))$ with $z_i(t)\in \mathds{R}^{n-\nu_i}$, $i\in \mathcal{V}$. Suppose Assumptions \ref{ass6} and \ref{ass0} hold. Then, for all $\gamma > {\gamma}_{\max}$, and $\tau\in (0, \tau_{0})$, the system in \eqref{compactformii} is exponentially stable at the origin for all $ h_k\in (0, \tau]$ over $\mathds{N}$,
where
\begin{align}\label{barhbargamm}
 {\gamma}_{\max}=\frac{ 2\theta_M\sup_{i\in \mathcal{V}}\|A_{ui}\|}{\lambda_{l}}\;\textnormal{and}\;
 \tau_{0}=\frac{c_1}{c_2}, 
  \end{align}
with \begin{align}\label{C1C2}c_1=& \frac{\gamma\lambda_{l}}{\theta_M}-2\sup\limits_{i\in \mathcal{V}}\|A_{ui}\|,\;\nonumber\\
 c_2= &\frac{\sup\nolimits_{i\in \mathcal{V}}(\|A_{ui}\|+\gamma\|V_{ui}\|^2\|\mathcal{{L}}\|)}{\gamma^{-1}\lambda_{L}^{-1}\theta_m}.\end{align}
 \end{lem}

\section{Main Results}\label{mainresults}
In this section, we present the design and  analysis of the proposed aperiodic sampled-data distributed observers.
\subsection{Aperiodic Sampled-Data Distributed Observers Design}\label{mainresults1}
The dynamics for the proposed linear distributed observer $\forall t\in [t_k,t_{k+1})$ is given by:
\begin{align}
\dot{\hat{x}}_i(t)=&A\hat{x}_i(t) + L_i(C_i\hat{x}_i(t_{k})-y_i(t_{k}))\nonumber\\
&+\gamma M_i\sum\nolimits_{j\in\mathcal{N}_i}{(\hat{x}_j(t_{k})-\hat{x}_i(t_{k}))},\;i\in \mathcal{V},\label{compensator}
\end{align}
 where $t_0$ is the initial time, $t_{k+1}-t_k=h_k$ over $k\in \mathds{N}$, $\gamma\geq \gamma_{\max}$ and $h_k\in (0,h_{\max})$ with {\myr $h_{\max}$ being some positive numbers to be determined, and $\gamma_{\max}$ is given in \eqref{barhbargamm}.} The matrices $L_i$ and $M_i$ are defined as follows:
 \begin{align}\label{generalV}
 L_i=V_i\left[
          \begin{array}{c}
            \textbf{0} \\
             L_{oi}\\
          \end{array}
        \right]\; \textnormal{and} \;M_i=V_i\left[
                           \begin{array}{cc}
                             I_{n-v_i}& \textbf{0} \\
                             \textbf{0} &\textbf{0}\\
                           \end{array}
                         \right]V_i^T,
\end{align}
with $L_{oi}\in \mathds{R}^{v_i\times p_i}$ chosen such that $\bar{A}_i\triangleq A_{oi}+L_{oi}C_{oi}$ is Hurwitz, for $i\in \mathcal{V}$.

The \emph{estimated allowable sampling bound} $h_{\max}$ is calculated following Algorithm \ref{algoth}. The computation requires the definition of the following function: 
 \begin{align}\label{Firstsampling}
 \tau_{1}(\chi, \kappa):=\left\{\begin{array}{cc}
                                        \frac{1}{\kappa r}\arctan(r), & \textnormal{if}\; \chi>\kappa; \\
                                        \frac{1}{\kappa},   & \textnormal{if}\; \chi=\kappa; \\
                                        \frac{1}{\kappa r}\operatorname{arctanh}(r), & \textnormal{if}\; \chi<\kappa;
                                     \end{array}
 \right.
 \end{align}
 where $r=\sqrt{|\frac{\chi^2}{\kappa^2}-1|}$, $\kappa=\sup\limits_{i\in \mathcal{V}}\|L_{oi}C_{oi}\|$ and $\chi\geq \chi_{\max}$ with $\chi_{\max}=\max\{\chi_{1},\cdots,\chi_{N}\}$ where each $\chi_i$ is given as:
 \begin{equation}\label{ALCHINFI}
 \chi_i=\|\bar{A}_i^T(sI-\bar{A}_i)^{-1}L_{oi}C_{oi}\|_{\infty}.
 \end{equation}
\begin{algorithm}
\caption{\myr An estimated allowable sampling bound $h_{\max}$}\label{algoth}
\begin{algorithmic}
   \State \textbf{1.} Select $\Theta=\diag(\theta_1,\cdots,\theta_N)$ such that $\Theta \mathcal{L}+\mathcal{L}^T\Theta$ is positive semi-definite matrix.
   \State \textbf{2.} Choose positive constant $\gamma>\gamma_{\max}$ using \eqref{barhbargamm}.
   \State \textbf{3.} Compute $\tau_{0}=\frac{c_1}{c_2}$ given in \eqref{barhbargamm} and \eqref{C1C2}.
    \State \textbf{4.} Select $\chi\geq\chi_{\max}$ using \eqref{ALCHINFI}.
    \State \textbf{5.} Calculate $\tau_{1}(\chi, \kappa)$ using \eqref{Firstsampling}.
   \State \textbf{6.} Let $h_{\max}=\min\{\tau_{1}(\chi, \kappa), \tau_{0}\}$.
  %
\end{algorithmic}
\end{algorithm}
\begin{rem} The jointly observable distributed observer for the continuous-time case was considered in \cite{kim2019completely}. In terms of our notation, the convergence of the observer proposed in \cite{kim2019completely} can be guaranteed if $\gamma$ is chosen such that
\begin{align*}\gamma\geq\frac{ 2\theta_M\sup_{i\in \mathcal{V}}\|A_{ui}\|+\psi}{\lambda_{l}},\end{align*}
where $\psi$ is larger than the maximum of all the real parts of the eigenvalues of $A_{oi}+L_{oi}C_{oi}$. As $h_k\rightarrow0$, the sampled-data distributed observer in \eqref{compensator} reduces to a  continuous-time observer. For the design proposed in this study, the constant $\gamma_{\max}$ as defined in \eqref{barhbargamm} can be chosen smaller than $\displaystyle{\frac{ 2\theta_M\sup_{i\in \mathcal{V}}\|A_{ui}\|+\psi}{\lambda_{l}}}$. 

{\myr It is important to note that the required gain $\gamma_{\max}$ and estimated allowable sampling bound $\tau_0$ in \eqref{barhbargamm} relies on centralized properties of the system such as $\theta_M$, $\sup_{i\in \mathcal{V}}\|A_{ui}\|$ and $\lambda_{l}$, which are graph related parameters, parameters of the unobservable parts of the system to be observed and the minimum eigenvalues of distributed state estimation induced matrix $V_{u}^T(\mathcal{\hat{L}}\otimes I_n)V_{u}$. Therefore, if there are some uncertainties arising from the system \eqref{leader} and the network, the robustness issues will impact the choice of the gain $\gamma_{\max}$ and the estimated allowable sampling bound $\tau_0$, increasing the complexity of the distributed state estimation problem based on sampled data. Some work in continuous-time cases has been considered in \cite{wang2022robust,yang2022state,cao2023distributed} for systems subject to unknown inputs and external disturbances. Further research is required to address these situations. 
It should also be noted that one adaptive approach has been proposed in \cite{kim2019completely} for the adaptive estimation of the gain $\gamma_{\max}$ in the continuous-time case.  }
\end{rem}

\subsection{Convergence Analysis}

For $i\in \mathcal{V}$, let $\tilde{x}_i(t)=\hat{x}_i(t)-x(t)$ be the estimation error of the $i$-th observer at time instant $t$. Then, $\forall t\in [t_k,t_{k+1})$, we have
\begin{align}\label{decom1}
\dot{\tilde{x}}_i(t)=&A\tilde{x}_i(t)+ L_iC_i\tilde{x}_i(t_{k})\nonumber\\
& \hspace{0.385in}+\gamma M_i\sum\nolimits_{j\in\mathcal{N}_i}{(\tilde{x}_j(t_{k})-\tilde{x}_i(t_{k}))}\nonumber\\
=&A\tilde{x}_i(t)+ L_iC_i\tilde{x}_i(t_{k})-\gamma M_i\sum\nolimits_{j=1}^{N} l_{ij}{\tilde{x}_j(t_{k})},
\end{align}
where $l_{ij}$ is the $(i,j)$-th entry of the Laplacian matrix $\mathcal{L}$.
Let $\tilde{x}_{oi}=V_{oi}^T\tilde{x}_i$ and $\tilde{x}_{ui}=V_{ui}^T\tilde{x}_i$, for $i\in \mathcal{V}$. Then, we have the following system from \eqref{decom} and \eqref{decom1}, $\forall t\in [t_k,t_{k+1})$,
\begin{subequations}\label{decom2}\begin{align}
\dot{\tilde{x}}_{ui}(t)=&A_{ui}\tilde{x}_{ui}(t)+A_{ri}\tilde{x}_{oi}(t)\nonumber\\
&-\gamma V_{ui}^T\sum\limits_{j=1}^{N} l_{ij}\big[{V_{uj}\tilde{x}_{uj}(t_{k})+V_{oj}\tilde{x}_{oj}(t_{k})}\big],\label{decom2b}\\
\dot{\tilde{x}}_{oi}(t)=&A_{oi}\tilde{x}_{oi}(t)+L_{oi}C_{oi}\tilde{x}_{oi}(t_k),\;i\in \mathcal{V}.\label{decom2a}
\end{align}
\end{subequations}
Let $\tilde{x}_{u}=\col(\tilde{x}_{u1},\cdots,\tilde{x}_{uN})$, $\tilde{x}_{o}=\col(\tilde{x}_{o1},\cdots,\tilde{x}_{oN})$ and $L_{o}=\textnormal{diag}({L_{o1}},\cdots,{L_{oN}})$. Then, the system in \eqref{decom2} can be put into the following compact form, $\forall t\in [t_k,t_{k+1})$,
\begin{subequations}\label{compactform}
\begin{align}
\dot{\tilde{x}}_{u}(t)=&A_{u}\tilde{x}_{u}(t)+A_r\tilde{x}_{o}(t)\nonumber\\
&-\gamma V_{u}^T\left(\mathcal{L}\otimes I_n\right)\big[V_{u}\tilde{x}_{u}(t_k)+V_{o}\tilde{x}_{o}(t_{k})\big],\label{compactforma}\\
\dot{\tilde{x}}_{o}(t)=&A_{o}\tilde{x}_{o}(t)+L_{o}C_{o}\tilde{x}_{o}(t_k),\label{compactformb}
\end{align}\end{subequations}
where $A_{u}$, $A_r$, $A_{o}$, $C_{o}$ and $V_{u}$ are defined in Remark \ref{diagObser}.

The first step of the analysis of convergence of system \eqref{compactform} is to establish the stability of the system  \eqref{compactformb}. In the following lemma, an emulation-based approach as proposed in \cite{nesic2009explicit} is used to analyze the stability properties of \eqref{compactformb}.

\begin{lem}\label{discrettimeexo} For the system in \eqref{compactformb}, suppose Assumption \ref{ass1} holds. Choose $\kappa$, $\chi_{\max}$, $\chi$ and $\tau_{1}$ according to \eqref{Firstsampling} and \eqref{ALCHINFI}. Then, for all $\chi \geq \chi_{\max}$ and $\tau\in (0, \tau_{1})$, the system in \eqref{compactformb} is exponentially stable at the origin for all $ h_k\in (0, \tau]$, $k\in \mathds{N}$.
\end{lem}
\begin{proof} Let $e(t)=\tilde{x}_{o}(t_k)-\tilde{x}_{o}(t)$ be the sampling-induced error, for any $t\in [t_k, t_{k+1})$ over $k\in \mathds{N}$ and $i\in \mathcal{V}$. Then, the dynamics in \eqref{compactformb} can be rewritten in the following manner:
\begin{align}\label{hbrid}
    \left\{       \begin{array}{cc}
                   \dot{\tilde{x}}_{o}(t)=f(\tilde{x}_{o}(t), e(t)),&\forall t\in [t_k, t_{k+1});\\
           \dot{e}(t) =g(\tilde{x}_{o}(t), e(t)),&\forall t\in [t_k, t_{k+1}); \\
         e(t_k^{+})=\textbf{0},& k\in \mathds{N};
           \end{array}\right.
\end{align}
where {\myr $f(\tilde{x}_{o}, e)=\bar{A}_d\tilde{x}_{o}+L_{o}C_{o}e$ and $g(\tilde{x}_{o}, e)=-f(\tilde{x}_{o}, e)$ with 
$\bar{A}_d=\diag(\bar{A}_1,\cdots, \bar{A}_N)$ and $\bar{A}_i=A_{oi}+L_{oi}C_{oi}$, $i\in \mathcal{V}$.}

Let $\tau_{1}(\chi, \kappa)$ be defined following \eqref{Firstsampling}. The dynamics of \eqref{hbrid} with $h_k\in[\epsilon, \tau_{1}]$ can be modeled as a hybrid system of the form:
\begin{equation}\label{hybrimodel}
  \begin{cases}
           \left. \begin{array}{c}
          \dot{\tilde{x}}_{o}=f(\tilde{x}_{o}, e)\\
           \dot{e} =g(\tilde{x}_{o}, e) \\
          \dot{\tau}=1 \\
         \end{array}\right\}\;{\tau}\in[0, \tau_{1}],\;\textnormal{flow dynamics};\\
               \left. \begin{array}{c}
          \tilde{x}_{o}^{+}=\tilde{x}_{o}\\
           {e}^{+} =\textbf{0} \\
          {\tau}^{+}=0 \\
         \end{array}\right\}\;{\tau}\in[\epsilon, \infty],\;\textnormal{jump dynamics}\footnotemark,     \end{cases}
\end{equation}
where $\tau \in \mathds{R}_{+}$ is a clock state, $\epsilon$ is an arbitrary small positive number.  
\footnotetext{$x^{+}$ denotes the state of a hybrid system after a jump.}
For, $\lambda\in(0, 1)$, let $\phi:[0,\tau_{1}]\rightarrow \mathds{R}$ be the solution of the following differential equation:
\begin{equation*}\dot{\phi}=-2\kappa \phi- \chi (\phi^2+1),\;\phi(0)=\lambda^{-1}.\end{equation*}
According to \cite{carnevale2007lyapunov}, it can be shown that $\phi\in [\lambda, \lambda^{-1}]$. For $i\in \mathcal{V}$, $\bar{A}_i$ is Hurwitz, for any positive number $\chi\geq\chi_{\max}$, where $\chi_{\max}$ is defined in \eqref{ALCHINFI}. As a result, there exists a positive definite matrix $P_i\in \mathds{R}^{\nu_i\times \nu_i}$ that satisfies the matrix inequality:
 \begin{align}\label{PQALC}
  \left[
    \begin{array}{cc}
       \bar{A}_i^TP_i+P_i\bar{A}_i+\frac{1}{\chi}\bar{A}_i^T\bar{A}_i & P_iL_{oi}C_{oi} \\
      C_{oi}^TL_{oi}^TP_i & -\chi I_{v_i}\\
    \end{array}
  \right] < 0.
  \end{align}
We pose a candidate Lyapunov function as follows:
$$U(\tau,\tilde{x}_{o}, e)= \tilde{x}_{o}^TP_d\tilde{x}_{o}+\phi(\tau)e^Te,$$
where $P_d=\diag(P_1,\cdots,P_{N})$ {\myr in which each $P_i$ is a positive definite symmetric solution} of the matrix inequality \eqref{PQALC}.
On the jump domain, \eqref{hybrimodel}, it is noted that
\begin{align*}
U(\tau^+,\tilde{x}_{o}^+, e^+)=&(\tilde{x}_{o}^+)^TP_d\tilde{x}_{o}^++\phi(\tau^+)(e^+)^Te^+\\
=&\tilde{x}_{o}^TP_d\tilde{x}_{o}\leq U(\tau,\tilde{x}_{o}, e).
\end{align*}
In addition, from the flow dynamics in \eqref{hybrimodel}, we obtain the following inequalities for the time derivative of the Lyapunov function $U(\tau,\tilde{x}_{o}, e)$:
\begin{align*}
\dot{U}
= \,\, &\tilde{x}_{o}^T\big[P_d\bar{A}_d+\bar{A}_d^TP_d\big]\tilde{x}_{o}+2\tilde{x}_{o}^TP_dL_{o}C_{o}e\\
& +\big[-2\kappa \phi(\tau)e^Te- \chi (\phi^2(\tau)+1)e^Te\big]\\
&+2\phi(\tau)e^T\big[-\bar{A}_d\tilde{x}_{o}-L_{o}C_{o}e\big]\\
\leq \,\, &\tilde{x}_{o}^T\big[P_d\bar{A}_d+\bar{A}_d^TP_d+\frac{1}{\chi}\bar{A}_d^T\bar{A}_d\big]\tilde{x}_{o}\\
&+2\tilde{x}_{o}^TP_dL_{o}C_{o}e+\chi\phi^2(\tau)e^Te-2\phi(\tau)e^TL_{o}C_{o}e\\
&+\big[-2\kappa \phi(\tau)e^Te- \chi (\phi^2(\tau)+1)e^Te\big]\\
\leq \,\, &\tilde{x}_{o}^T\big[P_d\bar{A}_d+\bar{A}_d^TP_d+\frac{1}{\chi}\bar{A}_d^T\bar{A}_d\big]\tilde{x}_{o}\\
&+2\tilde{x}_{o}^TP_dL_{o}C_{o}e- \chi e^Te\\
\leq & -\rho^{*}U,
\end{align*}
where $\rho^{*}$ is a positive constant. Since $\bar{A}_d$ is Hurwitz, it follows that the system $\dot{\tilde{x}}_{o}=f(\tilde{x}_{o}, 0)$ is exponentially stable. It then follows from Theorem 2 of \cite{nesic2009explicit} that {\myr the set $\{\tilde{x}_o=0,e=0,\tau\in[0,\tau_1]\}$ is exponentially stable for system \eqref{hybrimodel}}. This further implies {\myr system \eqref{compactformb} is exponentially stable at the origin for all $ h_k\in (0, \tau]$ over $k\in \mathds{N}$.} This completes the proof.
\end{proof}


To analyze the stability of the system in \eqref{compactform}, we discretize the continuous-time system \eqref{decom2} into a discrete-time system form. We employ the step-invariant transformation discretization technique found in \cite{chen2012optimal}  to transform the continuous-time system \eqref{compactform} to the following  time-varying discrete-time system:
\begin{subequations}\label{souludecomlamda}\begin{align}
{\tilde{x}}_{u}(t_{k+1})=&\Lambda(h_k)\tilde{x}_{u}(t_k)+g(t_k),\label{souludecom2}\\
{\tilde{x}}_{oi}(t_{k+1})=&\Big[e^{A_{oi}h_k}+\int_{0}^{h_k}e^{A_{oi}\tau}L_{oi} C_{oi}d\tau \Big]\tilde{x}_{oi}(t_k),\label{discrete}
\end{align}\end{subequations}
where \begin{align*}
\Lambda(h_k)=&e^{A_{u}h_k}-\gamma \int_{0}^{h_k}e^{A_{u}\tau}d\tau V_{u}^T\left(\mathcal{L}\otimes I_n\right)V_{u},\\
g(t_k)=&-\gamma \int_{0}^{h_k}e^{A_{u}\tau}d\tau V_{u}^T\left(\mathcal{L}\otimes I_n\right) V_{o}\tilde{x}_{o}(t_{k})\\
&+\int_{0}^{h_k}e^{A_{u}\tau}A_{r}\tilde{x}_{o}(t_k+\tau)d\tau.\end{align*} Then, we have the following results.
\begin{Corollary} \label{corr} Consider system \eqref{souludecomlamda} and let Assumptions \ref{ass0} and \ref{ass1} be met. Furthermore, choose ${\gamma}_{\max}$ and $\tau_{0}$ according to \eqref{barhbargamm}. Then, for all $\gamma \geq  {\gamma}_{\max}$ and $\tau\in (0, \tau_{0})$, the system \eqref{souludecomlamda} is exponentially stable at the origin for all $ h_k\in (0, \tau]$ over $k\in \mathds{N}$, provided that {\myr $\tilde{x}_{o}(t)=0$ for all $t\geq 0$}.
\end{Corollary}
\begin{proof} As $\tilde{x}_{o}(t)=0$ for all $t\geq0$, the system in \eqref{souludecom2} reduces to
\begin{equation}\label{sampleddata}
{\tilde{x}}_{u}(t_{k+1})=\Lambda(h_k)\tilde{x}_{u}(t_k).
\end{equation}
By Lemma \ref{lemmaexp}, under Assumptions \ref{ass0} and \ref{ass1}, for all $\gamma \geq {\gamma}_{\max}$ and $\tau\in (0, \tau_{0})$, the system in \eqref{compactform} is exponentially stable at the origin for all $ h_k\in (0, \tau]$, $k\in \mathds{N}$, provided that $\tilde{x}_{o}(t)=\textbf{0}$ for all $t\geq 0$.
Hence, the system in \eqref{sampleddata} is exponentially stable for all $\gamma \geq \gamma_{\max}$ and $ h_k\in (0, \tau]$ over $k\in \mathds{N}$.
\end{proof}

\begin{thm}\label{discrettimeexo}  Consider system \eqref{compactform}. Let Assumptions \ref{ass0} and \ref{ass1} be satisfied. Let ${\gamma}_{\max}$ and $h_{\max}$ be chosen according to \eqref{barhbargamm} and Algorithm \ref{algoth}. Then, for all $\gamma \geq  {\gamma}_{\max}$ and $\tau\in (0, h_{\max})$, the origin is an exponentially stable equilibrium of system \eqref{compactform} for all $ h_k\in (0, \tau]$ over $k\in \mathds{N}$ such that
$$ \lim_{t\rightarrow\infty}(\hat{x}_i(t)-x(t))=\textbf{0},\;i\in \mathcal{V},$$ for any $x(0)\in \mathds{R}^{n}$ and $\hat{x}_i(0)\in \mathds{R}^{n}$.
\end{thm}
\begin{proof} Under Assumptions \ref{ass0} and \ref{ass1}, choose $ \gamma_{\max}$ and $h_{\max}$ according to \eqref{barhbargamm} and Algorithm \ref{algoth}. Then, for all $\gamma \geq  \gamma_{\max}$ and $\tau\in (0, h_{\max})$, the system in \eqref{sampleddata} is exponentially stable at the origin from Corollary \ref{corr}. By the \emph{Converse Lyapunov Theorem} (\citep*{stein1952some}, \citep*[Theorem~23.3]{rugh1996linear} and \citep*{bai1988averaging}), there exists a time-varying
symmetric matrix $P(k)$ over $k\in \mathds{N}$
 such that
\begin{subequations}
\begin{align}
\alpha_1I\leq P(k)\leq&\alpha_2I,\\
\Lambda^T(h_k)P({k+1})\Lambda(h_k)- P(k)\leq& -\alpha_3 I,
\end{align}
\end{subequations}
for some positive constants $\alpha_1$, $\alpha_2$ and $\alpha_3$. Choose the Lyapunov function for the system in \eqref{sampleddata} as follows: $$U(t_k)=\tilde{x}_{u}^T(t_k)P(k)\tilde{x}_{u}(t_k).$$  Then, along the trajectory of the system in \eqref{souludecom2}, we have
\begin{align}\label{differVdefi}
 U(&t_{k+1})-U(t_{k})\nonumber\\
=&\big[\Lambda(h_k)\tilde{x}_{u}(t_k)+g(t_k)\big]^TP({k+1})\big[\Lambda(h_k)\tilde{x}_{u}(t_k)+g(t_k)\big]\nonumber\\
&-\tilde{x}_{u}^T(t_k)P(k)\tilde{x}_{u}(t_k)\nonumber\\
\leq& -\alpha_3\|\tilde{x}_{u}(t_k)\|^2+\alpha_2\|g(t_k)\|^2+2\alpha_2\|g(t_k)\|\|\tilde{x}_{u}(t_k)\|\nonumber\\
\leq&-\frac{3\alpha_3}{4}\|\tilde{x}_{u}(t_k)\|^2+\frac{\alpha_2^2+4\alpha_2\alpha_3}{4\alpha_3}\|g(t_k)\|^2.
\end{align}
By Lemma \ref{discrettimeexo}, under Assumptions \ref{ass0} and \ref{ass1}, for all $\gamma \geq \gamma_{\max}$ and $\tau\in (0, \tau_{0})$, we have
$\lim_{t\rightarrow \infty}\tilde{x}_{oi}(t)=\textbf{0}$. As a result, the trajectories of the system are such that $\lim_{k\rightarrow \infty}\|g(t_k)\|=0$ exponentially and $\|g(t_k)\|$ is bounded over $\mathds{N}$.
The inequality \eqref{differVdefi} proves that system \eqref{souludecom2} is input-to-state stable with $\frac{\alpha_2^2+4\alpha_2\alpha_3}{4\alpha_3}\|g(t_k)\|^2$ as the input. By Lemma 3.8 in \cite{jiang2001input}, the system in \eqref{souludecom2} has the $\mathcal{K}$ asymptotic gain property. Hence, there exists a class $\mathcal{K}$ function $\beta(\cdot)$ such that, for any initial condition, the solution of \eqref{souludecom2} satisfies
$$\limsup_{k\rightarrow \infty}\|\tilde{x}_{u}(t_k)\|\leq \beta\left(\limsup_{k\rightarrow \infty}\frac{\alpha_2^2+4\alpha_2\alpha_3}{4\alpha_3}\|g(t_k)\|^2\right).$$
Therefore, $\lim_{k\rightarrow \infty}\frac{\alpha_2^2+4\alpha_2\alpha_3}{4\alpha_3}\|g(t_k)\|^2=0$ implies that $\lim_{k\rightarrow \infty}\tilde{x}_{u}(t_k)=\textbf{0}$ exponentially. It follows from the system in \eqref{compactforma} that, $\forall t\in [t_k,t_{k+1})$,
\begin{align}
\|{\tilde{x}}_{u}(t)\|&\leq\|{\tilde{x}}_{u}(t_k)\|e^{\|A_{u}\|\tau}\big(1+\gamma  \tau \|\mathcal{L}\|\|V_{u}\|^2\big) \nonumber\\
&+ \tau e^{\|A_{u}\|\tau} \big(\|A_r\|+\gamma \|\mathcal{L}\|\|V_{u}\|\|V_{o}\|\big)\|\tilde{x}_{o}(t_k)\|.\nonumber
\end{align}
The last inequality, along with the fact that  $\lim\limits_{k\rightarrow\infty}\tilde{x}_{u}(t_k)=\textbf{0}$ and $\lim\limits_{t\rightarrow\infty}\tilde{x}_{o}(t)=\textbf{0}$, proves that $\lim\limits_{t\rightarrow\infty}\tilde{x}_{u}(t)=\textbf{0}$. Using the identities $\tilde{x}_{oi}(t)=V_{oi}^T\tilde{x}_i(t)$ and $\tilde{x}_{ui}(t)=V_{ui}^T\tilde{x}_i(t)$, it follows from $\lim\limits_{t\rightarrow\infty}\tilde{x}_{ui}(t)=\textbf{0}$ and $\lim\limits_{t\rightarrow\infty}\tilde{x}_{oi}(t)=\textbf{0}$ that $\lim\limits_{t\rightarrow\infty}\tilde{x}_i(t)=\textbf{0},\;i\in \mathcal{V}$.
\end{proof}
\subsection{Application to Jointly Observable Tracking Problem}\label{mainresults2}
In this sub-section, we apply the distributed observer \eqref{decom1} to solve the \emph{Jointly Observable Tracking Problem} of linear multi-agent systems.
%

We now consider the design of control input as follows:
\begin{subequations}\label{control}\begin{align}
u_i(t) =&K_{\xi_i}\xi_i(t)+K_{x_i}\hat{x}_i(t),\label{controla}\\
\dot{\hat{x}}_i(t)=&A\hat{x}_i(t) + L_i(C_i\hat{x}_i(t_{k})-y_i(t_{k}))\nonumber\\
&+\gamma M_i\sum\nolimits_{j\in\mathcal{N}_i}{(\hat{x}_j(t_{k})-\hat{x}_i(t_{k}))},\;i\in \mathcal{V},
\end{align}\end{subequations}
 where $t_0=0$ is the initial time. We choose $\gamma \geq  \gamma_{\max}$ and $h_k\in (0, h_{\max})$ with  $ \gamma_{\max}$ and $h_{\max}$ using \eqref{barhbargamm} and Algorithm \ref{algoth}, $k\in \mathds{N}$, $K_{x_i}=U_i-K_{\xi_i}X_i$ and
$K_{\xi_i}$ is the feedback gain such that ${A}_i+{B}_i K_{\xi_i}$ is Hurwitz.
Then, we have the following theorem.

\begin{thm} Consider systems \eqref{leader} and \eqref{follower} and let Assumptions \ref{ass6} -- \ref{ass0} be satisfied. For any initial conditions $\xi_i(0)\in \mathds{R}^n$, $\hat{x}_i(0)\in \mathds{R}^n$ and $x(0)\in \mathds{R}^n$, Problem \ref{prob1} is solvable by the control law in \eqref{control} with $K_{x_i}=U_i-K_{\xi_i}X_i$ and $K_{\xi_i}$ chosen such that ${A}_i+{B}_i K_{\xi_i}$ is Hurwitz for $i\in \mathcal{V}$.
\end{thm}
\begin{proof} Let $\tilde{\xi}_i(t)=\xi_i(t)-X_ix(t)$ and $\tilde{u}_i(t)=u_i(t)-U_ix(t)$, for $i\in \mathcal{V}$. Then, we have
\begin{align}
\dot{\tilde{\xi}}_i(t)=&A_i\xi_i(t)+B_iu_i(t)-X_iAx(t)\notag\\
=& A_i{\tilde{\xi}}_i(t)+B_i \tilde{u}_i(t),\label{followererror}\\
e_i(t)=&F_i{\tilde{\xi}}_i(t)+D_i\tilde{u}_i(t),\\
\tilde{u}_i(t)=& K_{\xi_i}\tilde{\xi}_i(t)+K_{x_i}\tilde{x}_i(t),\;i\in \mathcal{V}.\label{followererror-control}
\end{align}
Upon substitution of the control law \eqref{followererror-control} into \eqref{followererror}, we obtain the following dynamics for the estimation error:
\begin{align}\label{followererror2}
\dot{\tilde{\xi}}_i(t)=&({A}_i+{B}_i K_{\xi_i}){\tilde{\xi}}_i(t)+B_i K_{x_i}\tilde{x}_i(t).
\end{align}
From Theorem \ref{discrettimeexo}, under Assumptions \ref{ass0}, and \ref{ass1}, there exists a positive $h_{\max}>0$ such that for any sampling periods $h_k\in (0, h_{\max})$ over $\mathds{N}$ and sufficiently large $\gamma$, $\tilde{x}_i(t)$ converges to zero exponentially as $t\rightarrow\infty$, for $i\in \mathcal{V}$. Moreover, ${A}_i+{B}_i K_{\xi_i}$ is Hurwitz. Thus, the system in \eqref{followererror2} can be viewed as a stable system with $-B_i K_{x_i}\tilde{x}_i(t)$ as the input, in which this input converges to zero as $t\rightarrow\infty$, for $i\in \mathcal{V}$. Hence, for any initial condition $\tilde{\xi}_i(0)$, $\lim_{t\rightarrow\infty}\tilde{\xi}_i(t)=\textbf{0}$, $i\in \mathcal{V}$.
\end{proof}
\section{Numerical Example}\label{numerexam}
\begin{figure}[htp]
\begin{center}
\begin{tikzpicture}[transform shape]
    \centering%
    \node (0) [circle, draw=red!20, fill=red!60, very thick, minimum size=7mm] {\textbf{0}};
    \node (1) [circle, right=of 0, draw=blue!20, fill=blue!60, very thick, minimum size=7mm] {\textbf{1}};
    \node (2) [circle, left=of 0, draw=blue!20, fill=blue!60, very thick, minimum size=7mm] {\textbf{2}};
    \node (3) [circle, above=of 0, draw=blue!20, fill=blue!60, very thick, minimum size=7mm] {\textbf{3}};
    \node (4) [circle, right=of 3, draw=blue!20, fill=blue!60, very thick, minimum size=7mm] {\textbf{4}};
    \node (5) [circle, left=of 3, draw=blue!20, fill=blue!60, very thick, minimum size=7mm] {\textbf{5}};
     \draw[red,thick,dashed,->] (0)--node[below]{$y_1$}(1);
     \draw[red,thick,dashed,->] (0)--node[below]{$y_2$}(2);
     \draw[red,thick,dashed,->] (0)--node[below]{$y_2$}(2);
      \draw[red,thick,dashed,->] (0)--node[left]{$y_3$}(3);
     \draw[blue,thick,dashed,->, left] (2) edge (5);
     \draw[blue,thick,dashed,->, left] (5) edge (3);
     \draw[blue,thick,dashed,->, left] (3) edge (2);
     \draw[blue,thick,dashed,->, left] (4) edge (1);
     \draw[blue,thick,dashed,->, left] (1) edge (3);
     \draw[blue,thick,dashed,->, left] (3) edge (4);
\end{tikzpicture}
\end{center}
\caption{ Communication topology $\bar{\mathcal{G}}$}
\label{fig1numex}
\end{figure}
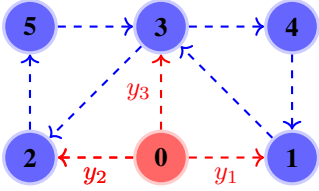
In this example, we consider a linear distributed system composed of an LTI system over the five node dynamics shown in Fig.\ref{fig1numex}. The dynamics of the LTI system \eqref{leader} with 
\begin{align}\nonumber
A=&\left[\begin{matrix}
               0 & 0.1 & 0\\
              -0.1 & 0 & 0\\
              0 & 0 & 0.1
             \end{matrix}
           \right],\nonumber\\ C^T =&\begin{matrix}  & C_1^T & C_2^T & C_3^T & C_4^T& C_5^T\\
\ldelim[{3}{0cm}& 1& 0 & 0& 0& 0 \rdelim]{3}{0cm}\\
&0&1 & 0 & 0 & 0\\
&0&0 & 1 & 0 & 0\\
& &  &   &   &
\end{matrix}.\nonumber 
\end{align}
The partition $y_1(t_k)$, $y_2(t_k)$ and $y_3(t_k)$ of the augmented output $y(t)=\col(y_1(t),y_2(t),y_{3}(t),y_4(t),y_5(t))\equiv Cx(t)$ are measured at time instant $t_k$ by the $1^{st}$, $2^{nd}$ and $3^{rd}$ agents, respectively,  as shown in Fig.\ref{fig1numex}. The system is also such that the $4^{th}$ and $5^{th}$ agents do not receive any direct measurement from the system to be observed. As a result, we set $y_4(t)=0$ and $y_5(t)=0$.
Furthermore, we can also see that, $(A, C)$ is observable, but none of the local pairs $(A, C_i)$ are observable.
The dynamics of the agent in \eqref{follower} with $D_i=0$, $Y_i=I_3$ is described by the matrices:
\begin{align}
A_i=&\left[\begin{matrix}
               0 & 0.1 & 0\\
              -0.1 & 0 & 0\\
              0 & 0 & 0.1
             \end{matrix}
           \right],\;B_i=\left[\begin{matrix}0\\1\\1 \end{matrix}\right], F_i=I_3. \nonumber
\end{align}
Next, we choose the following matrices based on Kalman's observability decomposition:
\begin{align}\label{Vtrueva}V_1=&\left[
      \begin{matrix}
         0& 0&1\\
         0& 1&0\\
        -1&0& 0\\
      \end{matrix}
    \right], \;V_2=\left[
      \begin{matrix}
        0   & -1   &  0\\
     0  &   0 &    1\\
    -1   &  0  &   0\\
      \end{matrix}
    \right],\nonumber\\
 V_3=& \left[
      \begin{matrix}
     -1 &    0 &    0\\
     0  &   1  &   0\\
     0  &   0  &   1\\
      \end{matrix}
    \right],\;V_4=I_3,\;V_5=I_3. \end{align}
It can be verified that the topology in Fig.\ref{fig1numex} satisfies Assumption \ref{ass0}. 
Let $\theta=\col(1,1,1,1,1)$ such that $\lambda_L= 5$ and $\lambda_l=1$.
Then, we use \eqref{barhbargamm} and \eqref{ALCHINFI} to calculate the required constants $\gamma_{\max}=0.2$, $\chi_{\max}=4.8016$ and $\kappa=4.4721$.

Let $\gamma=0.4$ and $\chi=4.802$.  Algorithm \ref{algoth} and Eq.\eqref{barhbargamm} yield $\tau_{0}=0.0822$, $\tau_{1}(\chi,\kappa)= 0.5721$ and $h_{\max}= 0.0822$.
We design a control law composed of \eqref{compensator} and \eqref{control} with the following parameters: $K_{\xi_i}^T=\col(22.6, 46.7,-49.5)$, $K_{x_i}=-K_{\xi_i}$, $L_{o1}=\col(-4 -2)$, $L_{o2}=\col(-4 -2)$, and $L_{o3}=-2.5$. Then, from \eqref{generalV} and \eqref{Vtrueva}, we have
\begin{align*}
M_1=&\left[
      \begin{matrix}
         0& 0&0\\
         0& 0&0\\
        0&0& 1\\
      \end{matrix}
    \right], \;M_2=\left[
      \begin{matrix}
        0   & 0   &  0\\
     0  &   0 &    0\\
    0   &  0  &   1\\
      \end{matrix}
    \right],\; M_3= \left[
      \begin{matrix}
     1 &    0 &    0\\
     0  &   1  &   0\\
     0  &   0  &   0\\
      \end{matrix}
    \right],\\
M_4=&I_3,\;M_5=I_3,\;L_1=\col(-2,-4,0),\;L_2=\col(4,2,0),\\
L_3=&\col(0,0,-2.5),\;L_4=\col(0,0,0),\;L_5=\col(0,0,0).
\end{align*}
A simulation is carried out with the following initial condition:
$x(0)=\col(1,2,3)$, $\hat{x}_i(0)=\col(0,0,0)$ and $\xi_i(0)=\col(0,0,0)$, for $i\in \mathcal{V}$.

\begin{figure}[ht]
  \centering\setlength{\unitlength}{0.75mm}
\epsfig{figure=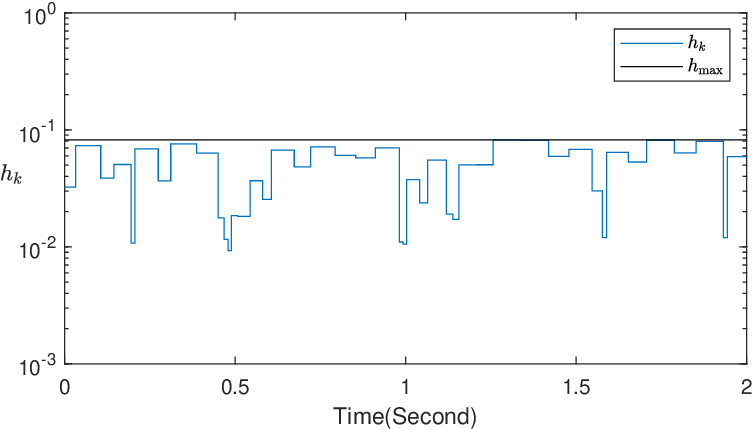,width=3.4in}
\caption{Sampling intervals as a function of the sampling time.}\label{figesmov2}
\end{figure}

\begin{figure}[ht]
  \centering\setlength{\unitlength}{0.75mm}
  \epsfig{figure=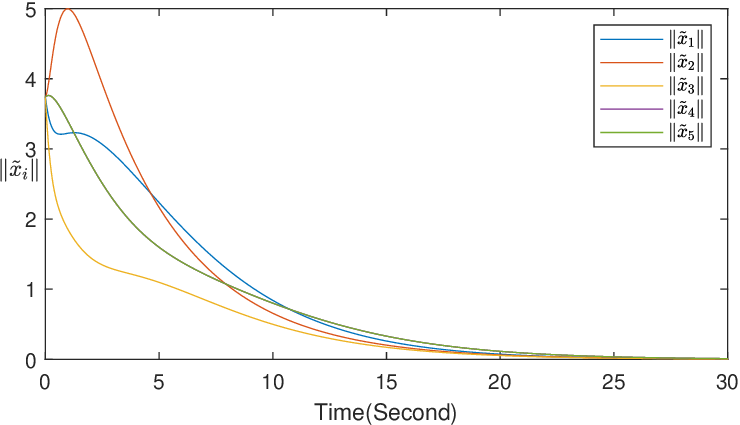,width=3.4in}
  \caption{Estimation errors of all agents, $i=1,\cdots,5$.}\label{figeses}
\end{figure}

\begin{figure}[ht]
  \centering\setlength{\unitlength}{0.75mm}
  \epsfig{figure=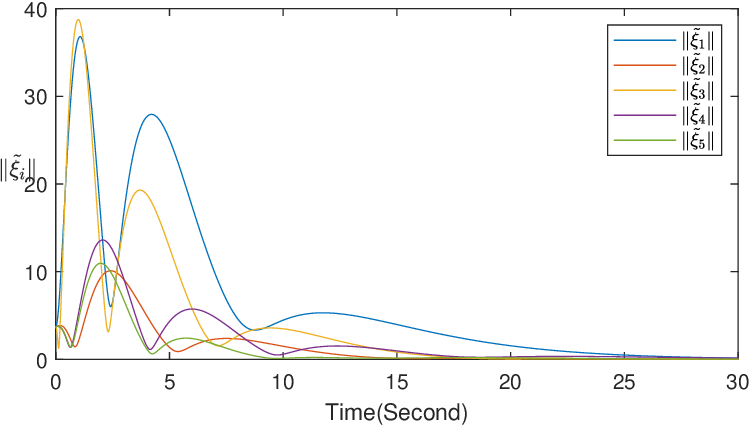,width=3.4in}
  \caption{Tracking errors of all node dynamics, $i=1,\cdots,5$.}\label{figesmov}
\end{figure}
The time-varying sampling periods as a function of sampling time are shown in Fig.\ref{figesmov2}.
In Fig.\ref{figeses}, the estimation error trajectories of all agents are given.
Finally, Fig.\ref{figesmov} shows the tracking error trajectories of all agents. The results confirm that all local estimation errors converge to zero, as expected.

\section{Conclusions}\label{conlu}
In this paper, a distributed state estimation problem
subject to a joint observability assumption has been investigated for sampled-data systems. An estimated allowable sampling bound for all agents is given to guarantee the convergence of the estimation error as long as the sampling periods of all agents are smaller than this upper bound. A distributed control law based on the distributed observer was synthesized to solve a cooperative tracking problem. The result relaxes the observability assumption required in \cite{su2011cooperative} and \cite{ding2013network} to a joint observability assumption. Although no agents can measure the entire output of the LTI system to be observed, each agent can asymptotically estimate the state of the system using only its aperiodic sampled measurements and its neighbors' estimation even when local measurements do not satisfy the Kalman observability condition. 

In contrast to the sampled-data approach based on the time-triggered strategy for sampling, the event-triggered strategy samples the continuous-time signals according to a prescribed design, or, adaptive triggering conditions, generating irregular observations and control updates. 
As a result, event-triggered techniques such as \cite{liu2018cooperative,zhang2023event,zhang2021hybrid} require further consideration as a mechanism to reduce unnecessary consumption of resources.

\begin{appendices}
\section{}\label{appendixa}
\begin{proof}
Define the following Lyapunov function for system \eqref{compactformii}
\begin{align}\label{lyapu}
U(z(t))=\sum\nolimits_{i=1}^{N}\theta_i\|z_{i}(t)\|^2
\end{align}
where $\theta_i$ is given in Algorithm \ref{algoth}. Then, $U(z(t))$ satisfies the following property
\begin{equation}\label{temp}
\theta_m\|z(t)\|^2\leq U(z(t))\leq\theta_M\|z(t)\|^2.
\end{equation}

For any $t\in [t_{k},t_{k+1})$, the time derivative of $U(t)$ along the trajectories of system (\ref{compactformii}) is given by:
\begin{align}
\dot{U}(t)
=&2\sum_{i=1}^{N}\theta_iz_{i}^T(t)A_{ui}z_{i}(t)-\gamma z^T(t)[V_{u}^T(\mathcal{\hat{L}}\otimes I_n)V_{u}]z(t)\nonumber\\
&-\gamma z^T(t)\big[V_{u}^T(\mathcal{\hat{L}}\otimes I_n)V_{u}\big][z(t_k)-z(t)]. 
\end{align}
Under Assumptions \ref{ass0} and \ref{ass1}, it follows from Lemma \ref{kimPD} that the matrix $V_{u}^T(\mathcal{\hat{L}}\otimes I_n)V_{u}$ is positive definite. Then, for any $t\in [t_{k},t_{k+1})$, we have
\begin{align}\label{dotV}
\dot{U}(t)\leq& 2\theta_M\sup_{i\in \mathcal{V}}\|A_{ui}\|\|z(t)\|^2-\gamma \lambda_{l}\|z(t)\|^2\nonumber\\
&+\gamma \lambda_{L}\|z(t)\| \|z(t)-z(t_k)\|.
\end{align}
From \eqref{compactformii} and \eqref{temp}, we have
\begin{align}\label{dotxnormal}
\|\dot{z}(t)&\|\leq\sup_{i\in \mathcal{V}}\|A_{ui}\|\|z(t)\|+\gamma\sup_{i\in \mathcal{V}}\|V_{ui}\|^2\|\mathcal{{L}}\|\|z(t_k)\|\nonumber\\
\leq&\frac{\sup\nolimits_{i\in \mathcal{V}}[\|A_{ui}\|+\gamma\|V_{ui}\|^2\|\mathcal{{L}}\|]}{\sqrt{\theta_m}}\sqrt{U_{M}(t_k)},
\end{align}
 for any $t\in [t_{k},t_{k+1})$ where $U_{M}(t_k)=\max\limits_{s\in[t_k,t_{k+1})}U(z(s))$. It is noted that, for any $t\in [t_{k},t_{k+1})$,
\begin{align}\label{xnormal}
\|&z(t)-z(t_k)\|\leq \int_{t_k}^{t}\|\dot{z}(s)\|ds\nonumber\\
\leq& \frac{\sup\nolimits_{i\in \mathcal{V}}[\|A_{ui}\|+\gamma\|V_{ui}\|^2\|\mathcal{{L}}\|]}{\sqrt{\theta_m}}\sqrt{U_{M}(t_k)}(t-t_k).
\end{align}
From \eqref{temp}, \eqref{dotV}, \eqref{dotxnormal} and \eqref{xnormal}, we obtain:
\begin{align}\label{dotV2}
\dot{U}(t)
\leq&  -c_1U(t)+ c_2 (t-t_k)\sqrt{U(t)}\sqrt{U_{M}(t_k)}\\
\leq&  -c_1U(t)+ c_2h_k  \sqrt{U(t)}\sqrt{U_{M}(t_k)},\;t\in [t_{k},t_{k+1}), \nonumber
\end{align}
where $c_1$ and $c_2$ are given in \eqref{C1C2}.

We then follow the arguments in \cite{liu2021sampled}. We first show that $U(t)=0$ for all $t\in[t_{k},t_{k+1})$ when $U(t_k)=0$ for some $k\in \mathds{N}$. If this is not true, then there exists a $t^*\in[t_{k},t_{k+1})$ such that $U(t^*)$ such that 
$\dot{U}(t^*)>0$, ${U}(t^*)\geq U(t)$ for all $t\in[t_{k}, t^{*})$. Hence $U_{M}(t_k)=W(t^*)$. Then, from \eqref{dotV2}, we have
$$\dot{U}(t^{*})\leq -(c_1- c_2\tau)U(t^{*})< 0$$
which contradicts the statement that $\dot{U}(t^*)>0$. 
Hence, $U(t)=0$ for all $t\in[t_{k},t_{k+1})$ when $U(t_k)=0$ for some $k\in \mathds{N}$.

We then show that when $U(t_k)>0$ for all $k\in \mathds{N}$ the following equation is satisfied
\begin{align}\label{sampleV}
\max\limits_{s\in[t_k,t_{k+1})}U(s)=U(t_k), ~~\forall~h_k\in(0, \tau],~~\tau<\tau_{0}.
\end{align}
If Eq.\eqref{sampleV} is not true, we can assume that there exists a time instant $t^{'}\in[t_k,t_{k+1}) $ such that $U(t^{'})> U(t_k)$. For any $\gamma> \gamma_{\max}$, it is noted from \eqref{dotV2} and \eqref{barhbargamm} that
\begin{align}\label{sampleV2}
\dot{U}(t_k)\leq& -c_1U(t_k)<0,~~\forall  z(t_k)\neq 0.
\end{align}
Thus $U(t)$ will decrease near the time instant $t_k$. Hence, there exists a time instant $t^{''}\in [t_k,t^{'}]$ such that
\begin{align}\label{relation}
(a)&~U(t^{''})=U(t_k),\nonumber\\
(b)&~\dot{U}(t^{''})>0,\nonumber\\
(c)&~U(t)\leq U(t^{''}),~~\forall t\in [t_k,t^{''}].
\end{align}
Then, equations \eqref{dotV2} and \eqref{relation} imply that 
\begin{align}\label{sampleV3}
\dot{U}(t^{''})\leq&-c_1U(t^{''})+ c_2 (t^{''}-t_k)\sqrt{U(t^{''})}\sqrt{U_{M}(t_k)}\nonumber\\
\leq& \big[\tau c_2-c_1\big]U(t^{''}).\nonumber
\end{align}
The fact that $\tau< \tau_{0}$ which leads to $\dot{U}(t^{''})\leq 0$, yields a contradiction of the second inequality in \eqref{relation}. Thus, equation \eqref{sampleV} must hold. From \eqref{dotV2}, we can write the following inequality:
\begin{equation}\label{dotV4}
\dot{U}(t)\leq -c_1U(t)+ c_2\tau\sqrt{U(t)}\sqrt{U(t_k)},\;\forall t\in [t_{k},t_{k+1}).
\end{equation}
Motivated by \cite{qian2012global}, let $\eta(t)=\sqrt{U(t)/U(t_k)}$. The time derivative of $\eta(t)$ on the time interval $[t_{k},t_{k+1})$ meets the following differential inequality:
\begin{align}\label{dotV5}
\dot{\eta}(t)
\leq&  -\frac{c_1}{2}\eta(t)+\frac{ c_2\tau}{2},\;\forall t\in [t_{k},t_{k+1}).
\end{align}
Using the comparison lemma \cite{khalil2002nonlinear}, we obtain from Eq.\eqref{dotV5} that:
\begin{align}
\eta(t)
\leq & e^{-\frac{c_1}{2}(t-t_k)}\Big(1-\frac{c_2\tau}{c_1}\Big)+\frac{c_2\tau}{c_1}~~~\forall t\in [t_{k},t_{k+1}).\nonumber
\end{align}
It is noted that $\eta(t_k)=1$ and $c_1-c_2\tau< 0$. As a result, we conclude that:
\begin{align}\label{rholess1}
\lim_{t\rightarrow t_{k+1}^{+}}\eta(t)=&\eta(t_{k+1}^{+})\leq  e^{-\frac{c_1}{2}(t_{k+1}-t_k)}\Big(1-\frac{c_2\tau}{c_1}\Big)+\frac{c_2\tau}{c_1}\nonumber\\
= & e^{-\frac{c_1}{2}h_k}\Big(1-\frac{c_2\tau}{c_1}\Big)+\frac{c_2\tau}{c_1}\triangleq \rho.
\end{align}
In addition, since $U(t)$ is continuous for all $t\geq0$,
 \begin{align*}
\lim_{t\rightarrow t_{k+1}^{+}}\eta(t)
=&\sqrt{U(t_{k+1})/U(t_k)},
\end{align*}
which together with Eq.\eqref{rholess1} yields
 \begin{align}\label{Utkrho}U(t_{k+1})\leq \rho^2U(t_k).\end{align}
Equation \eqref{rholess1} and $h_k\in (0, \tau_{0})$ lead to $0<\rho<1$. Therefore, we have shown that $U(t_k)$ converges to zero as $k$ tends to infinity (exponentially). This, together with Eq.\eqref{sampleV} and the fact that $U(t)=0$ for all $t\in[t_{k},t_{k+1})$ when $U(t_k)=0$ for some $k\in \mathds{N}$, along with \eqref{Utkrho} and \eqref{sampleV} further implies that $\lim_{t\rightarrow\infty}U(t)=0 $ exponentially. Finally, we conclude from from Eq.\eqref{temp} that system \eqref{compactformii} is exponentially stable at the origin.
\end{proof}
\end{appendices}

\noindent
\bibliographystyle{ieeetr}
\footnotesize
\bibliography{myref}

\begin{thebibliography}{10}

\bibitem{mitra2018distributed}
A.~Mitra and S.~Sundaram, ``Distributed observers for {LTI} systems,'' {\em
  IEEE Transactions on Automatic Control}, vol.~63, no.~11, pp.~3689--3704,
  2018.

\bibitem{wu2021design}
Y.~Wu, A.~Isidori, and R.~Lu, ``On the design of distributed observers for
  nonlinear systems,'' {\em IEEE Transactions on Automatic Control}, vol.~67,
  no.~7, pp.~3229 -- 3242, 2022.

\bibitem{olfati2007distributed}
R.~Olfati-Saber, ``Distributed {K}alman filtering for sensor networks,'' in
  {\em 2007 46th IEEE Conference on Decision and Control}, pp.~5492--5498,
  IEEE, 2007.

\bibitem{su2011cooperative}
Y.~Su and J.~Huang, ``{\myr Cooperative output regulation of linear multi-agent
  systems},'' {\em IEEE Transactions on Automatic Control}, vol.~57, no.~4,
  pp.~1062--1066, 2011.

\bibitem{su2012cooperative}
Y.~Su and J.~Huang, ``Cooperative output regulation with application to
  multi-agent consensus under switching network,'' {\em IEEE Transactions on
  Systems, Man, and Cybernetics, Part B (Cybernetics)}, vol.~42, no.~3,
  pp.~864--875, 2012.

\bibitem{park2016design}
S.~Park and N.~C. Martins, ``Design of distributed {LTI} observers for state
  omniscience,'' {\em IEEE Transactions on Automatic Control}, vol.~62, no.~2,
  pp.~561--576, 2016.

\bibitem{wang2017distributed}
L.~Wang and A.~S. Morse, ``A distributed observer for a time-invariant linear
  system,'' {\em IEEE Transactions on Automatic Control}, vol.~63, no.~7,
  pp.~2123--2130, 2017.

\bibitem{wang2019distributed}
L.~Wang, J.~Liu, A.~S. Morse, and B.~D. Anderson, ``A distributed observer for
  a discrete-time linear system,'' in {\em 2019 IEEE 58th Conference on
  Decision and Control (CDC)}, pp.~367--372, IEEE, 2019.

\bibitem{han2018simple}
W.~Han, H.~L. Trentelman, Z.~Wang, and Y.~Shen, ``A simple approach to
  distributed observer design for linear systems,'' {\em IEEE Transactions on
  Automatic Control}, vol.~64, no.~1, pp.~329--336, 2018.

\bibitem{han2018towards}
W.~Han, H.~L. Trentelman, Z.~Wang, and Y.~Shen, ``Towards a minimal order
  distributed observer for linear systems,'' {\em Systems \& Control Letters},
  vol.~114, pp.~59--65, 2018.

\bibitem{kim2019completely}
T.~Kim, C.~Lee, and H.~Shim, ``Completely decentralized design of distributed
  observer for linear systems,'' {\em IEEE Transactions on Automatic Control},
  vol.~65, no.~11, pp.~4664--4678, 2019.

\bibitem{zhang2022decentralized}
X.~Zhang and K.~Hengster-Movric, ``{\myr Decentralized design of distributed
  observers for {LTI} Systems},'' {\em IEEE Transactions on Automatic Control},
  2022.

\bibitem{wang2018adaptive}
S.~Wang and J.~Huang, ``Adaptive leader-following consensus for multiple
  {E}uler--{L}agrange systems with an uncertain leader system,'' {\em IEEE
  Transactions on Neural Networks and Learning Systems}, vol.~30, no.~7,
  pp.~2188--2196, 2018.

\bibitem{baldi2020distributed}
S.~Baldi, I.~A. Azzollini, and P.~A. Ioannou, ``A distributed indirect adaptive
  approach to cooperative tracking in networks of uncertain single-input
  single-output systems,'' {\em IEEE Transactions on Automatic Control},
  vol.~66, no.~10, pp.~4844--4851, 2020.

\bibitem{wang2022robust}
X.~Wang, H.~Su, F.~Zhang, and G.~Chen, ``{\myr A robust distributed interval
  observer for LTI systems},'' {\em IEEE Transactions on Automatic Control},
  vol.~68, no.~3, pp.~1337--1352, 2022.

\bibitem{yang2022state}
G.~Yang, A.~Barboni, H.~Rezaee, and T.~Parisini, ``{\myr State estimation using
  a network of distributed observers with unknown inputs},'' {\em Automatica},
  vol.~146, p.~110631, 2022.

\bibitem{cao2023distributed}
G.~Cao and J.~Wang, ``{\myr A distributed reduced-order unknown input
  observer},'' {\em Automatica}, vol.~155, p.~111174, 2023.

\bibitem{wang2023distributed}
S.~Wang and M.~Guay, ``{\myr Distributed state estimation for jointly
  observable linear systems over time-varying networks},'' {\em arXiv preprint
  arXiv:2302.12161}, 2023.

\bibitem{yang2023state}
G.~Yang, H.~Rezaee, A.~Alessandri, and T.~Parisini, ``{\myr State estimation
  using a network of distributed observers with switching communication
  topology},'' {\em Automatica}, vol.~147, p.~110690, 2023.

\bibitem{li2017robust}
Y.~Li, S.~Phillips, and R.~G. Sanfelice, ``Robust distributed estimation for
  linear systems under intermittent information,'' {\em IEEE Transactions on
  Automatic Control}, vol.~63, no.~4, pp.~973--988, 2017.

\bibitem{sferlazza2018time}
A.~Sferlazza, S.~Tarbouriech, and L.~Zaccarian, ``{\myr Time-varying
  sampled-data observer with asynchronous measurements},'' {\em IEEE
  Transactions on Automatic Control}, vol.~64, no.~2, pp.~869--876, 2018.

\bibitem{sferlazza2021state}
A.~Sferlazza, S.~Tarbouriech, and L.~Zaccarian, ``State observer with
  round-robin aperiodic sampled measurements with jitter,'' {\em Automatica},
  vol.~129, p.~Article 109573, 2021.

\bibitem{yu2011second}
W.~Yu, W.~X. Zheng, G.~Chen, W.~Ren, and J.~Cao, ``Second-order consensus in
  multi-agent dynamical systems with sampled position data,'' {\em Automatica},
  vol.~47, no.~7, pp.~1496--1503, 2011.

\bibitem{huang2016some}
N.~Huang, Z.~Duan, and G.~R. Chen, ``Some necessary and sufficient conditions
  for consensus of second-order multi-agent systems with sampled position
  data,'' {\em Automatica}, vol.~63, pp.~148--155, 2016.

\bibitem{wang2021cooperative}
S.~Wang, Z.~Shu, and T.~Chen, ``Distributed time- and event-triggered observers
  for linear systems: Non-pathological sampling and inter-event dynamics,''
  {\em arXiv preprint arXiv:2105.02200}, 2021.

\bibitem{ding2013network}
L.~Ding, Q.-L. Han, and G.~Guo, ``Network-based leader-following consensus for
  distributed multi-agent systems,'' {\em Automatica}, vol.~49, no.~7,
  pp.~2281--2286, 2013.

\bibitem{godsil2013algebraic}
C.~Godsil and G.~F. Royle, {\em Algebraic Graph Theory}, vol.~207.
\newblock New York: Springer-Verlag: Springer Science \& Business Media, 2013.

\bibitem{laila2002open}
D.~S. Laila, D.~Ne{\v{s}}i{\'c}, and A.~R. Teel, ``Open-and closed-loop
  dissipation inequalities under sampling and controller emulation,'' {\em
  European Journal of Control}, vol.~8, no.~2, pp.~109--125, 2002.

\bibitem{karafyllis2007small}
I.~Karafyllis and Z.-P. Jiang, ``A small-gain theorem for a wide class of
  feedback systems with control applications,'' {\em SIAM Journal on Control
  and Optimization}, vol.~46, no.~4, pp.~1483--1517, 2007.

\bibitem{nesic2009explicit}
D.~Nesic, A.~R. Teel, and D.~Carnevale, ``Explicit computation of the sampling
  period in emulation of controllers for nonlinear sampled-data systems,'' {\em
  IEEE Transactions on Automatic Control}, vol.~54, no.~3, pp.~619--624, 2009.

\bibitem{oishi2010stability}
Y.~Oishi and H.~Fujioka, ``Stability and stabilization of aperiodic
  sampled-data control systems using robust linear matrix inequalities,'' {\em
  Automatica}, vol.~46, no.~8, pp.~1327--1333, 2010.

\bibitem{zhang2015constructing}
H.~Zhang, Z.~Li, Z.~Qu, and F.~L. Lewis, ``On constructing lyapunov functions
  for multi-agent systems,'' {\em Automatica}, vol.~58, pp.~39--42, 2015.

\bibitem{carnevale2007lyapunov}
D.~Carnevale, A.~R. Teel, and D.~Nesic, ``A {L}yapunov proof of an improved
  maximum allowable transfer interval for networked control systems,'' {\em
  IEEE Transactions on Automatic Control}, vol.~52, no.~5, pp.~892--897, 2007.

\bibitem{chen2012optimal}
T.~Chen and B.~A. Francis, {\em Optimal Sampled-Data Control Systems}.
\newblock New York: Springer-Verlag: Springer Science \& Business Media, 1995.

\bibitem{stein1952some}
P.~Stein, ``Some general theorems on iterants,'' {\em Journal of Research of
  the National Bureau of Standards}, vol.~48, no.~1, pp.~82--83, 1952.

\bibitem{rugh1996linear}
W.~J. Rugh, {\em Linear System Theory, 2nd ed.}
\newblock Upper Saddle River, NJ, USA: Prentice-Hall, Inc., 1996.

\bibitem{bai1988averaging}
E.~Bai, L.~Fu, and S.~S. Sastry, ``Averaging analysis for discrete-time and
  sampled-data adaptive systems,'' {\em IEEE Transactions on Circuits and
  Systems}, vol.~35, no.~2, pp.~137--148, 1988.

\bibitem{jiang2001input}
Z.-P. Jiang and Y.~Wang, ``Input-to-state stability for discrete-time nonlinear
  systems,'' {\em Automatica}, vol.~37, no.~6, pp.~857--869, 2001.

\bibitem{liu2018cooperative}
W.~Liu and J.~Huang, ``{\myr Cooperative global robust output regulation for a
  class of nonlinear multi-agent systems by distributed event-triggered
  control},'' {\em Automatica}, vol.~93, pp.~138--148, 2018.

\bibitem{zhang2023event}
K.~Zhang, E.~Braverman, and B.~Gharesifard, ``{\myr Event-triggered control for
  discrete-time delay systems},'' {\em Automatica}, vol.~147, p.~110688, 2023.

\bibitem{zhang2021hybrid}
K.~Zhang and B.~Gharesifard, ``{\myr Hybrid event-triggered and impulsive
  control for time-delay systems},'' {\em Nonlinear Analysis: Hybrid Systems},
  vol.~43, p.~101109, 2021.

\bibitem{liu2021sampled}
W.~Liu and J.~Huang, ``{\myr Sampled-data cooperative output regulation of
  linear multi-agent systems},'' {\em International Journal of Robust and
  Nonlinear Control}, vol.~31, no.~10, pp.~4805--4822, 2021.

\bibitem{qian2012global}
C.~Qian and H.~Du, ``Global output feedback stabilization of a class of
  nonlinear systems via linear sampled-data control,'' {\em IEEE Transactions
  on Automatic Control}, vol.~57, no.~11, pp.~2934--2939, 2012.

\bibitem{khalil2002nonlinear}
H.~K. Khalil, {\em Nonlinear Systems, Third Edition}.
\newblock Upper Saddle River, NJ: Patience Hall, 2002.

\end{thebibliography}

\begin{IEEEbiography}[{\includegraphics[width=1in,height=1.25in,clip,keepaspectratio]{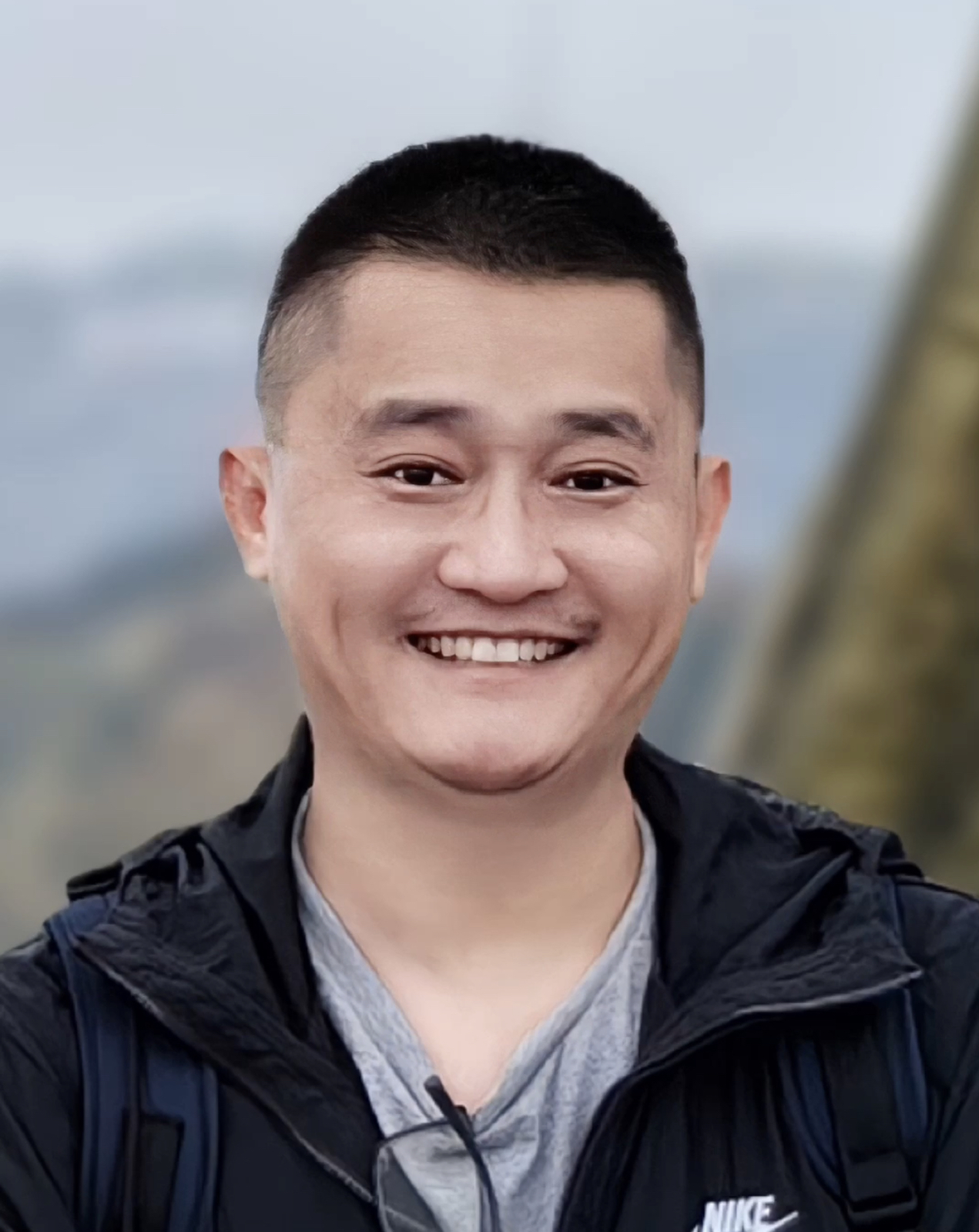}}]{Shimin Wang}
received the B.Sci. degree in Mathematics and Applied Mathematics and M.Eng. degree in Control Science and Control Engineering from Harbin Engineering University, Harbin, China, in 2011 and 2014, respectively. He then received the Ph.D. degree in Mechanical and Automation Engineering from The Chinese University of Hong Kong, Hong Kong, SAR, China, in 2019. 

He was a recipient of the NSERC Postdoctoral Fellowship award in 2022. From 2014 to 2015, he was an assistant engineer at the Jiangsu Automation Research Institute, China State Shipbuilding Corporation Limited. From 2019 to 2023, he held post-doctoral positions in the Department of Electrical and Computer Engineering at the University of Alberta, Canada and the Department of Chemical Engineering at Queen’s University, Canada, respectively. He is now a postdoctoral associate at the Department of Chemical Engineering at Massachusetts Institute of Technology, USA. 
\end{IEEEbiography}

\begin{IEEEbiography}[{\includegraphics[height=1.25in, clip,keepaspectratio]{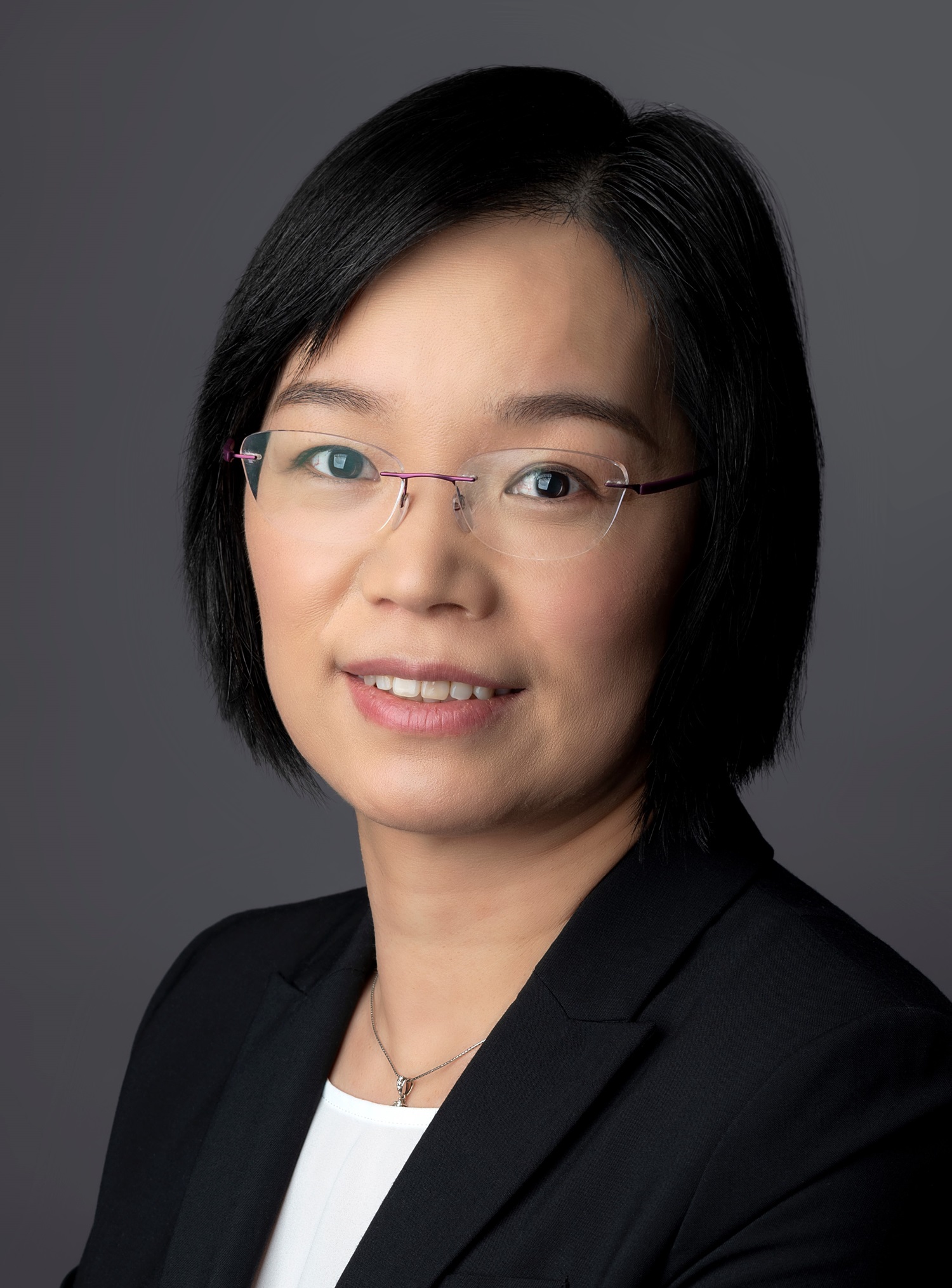}}]{Ya-Jun Pan}is a Professor in the Dept. of Mechanical Engineering at Dalhousie University, Canada. She received the B.E. degree in Mechanical Engineering from Yanshan University (1996), the M.E. degree in Mechanical Engineering from Zhejiang University (1999), and the Ph.D. degree in Electrical and Computer Engineering from the National University of Singapore (2003). She held post-doctoral positions at CNRS in the Laboratoire d’Automatique de Grenoble in France and the Dept. of Electrical and Computer Engineering at the University of Alberta in Canada, respectively. Her research interests are robust nonlinear control, cyber-physical systems, intelligent transportation systems, haptics, and collaborative multiple robotic systems.

She has served as Senior Editor and Technical Editor for IEEE/ASME Trans. on Mechatronics, Associate Editor for IEEE Trans. on Cybernetics, IEEE Trans. on Industrial Informatics, IEEE Industrial Electronics Magazine, and IEEE Trans. on Industrial Electronics. She is a Fellow of the Canadian Academy of Engineering (CAE), a Fellow of the Engineering Institute of Canada (EIC), a Fellow of ASME, a Fellow of CSME, a Senior Member of IEEE, and a registered Professional Engineer in Nova Scotia, Canada.
  \end{IEEEbiography}

\begin{IEEEbiography}[{\includegraphics[height=1.25in, clip,keepaspectratio]{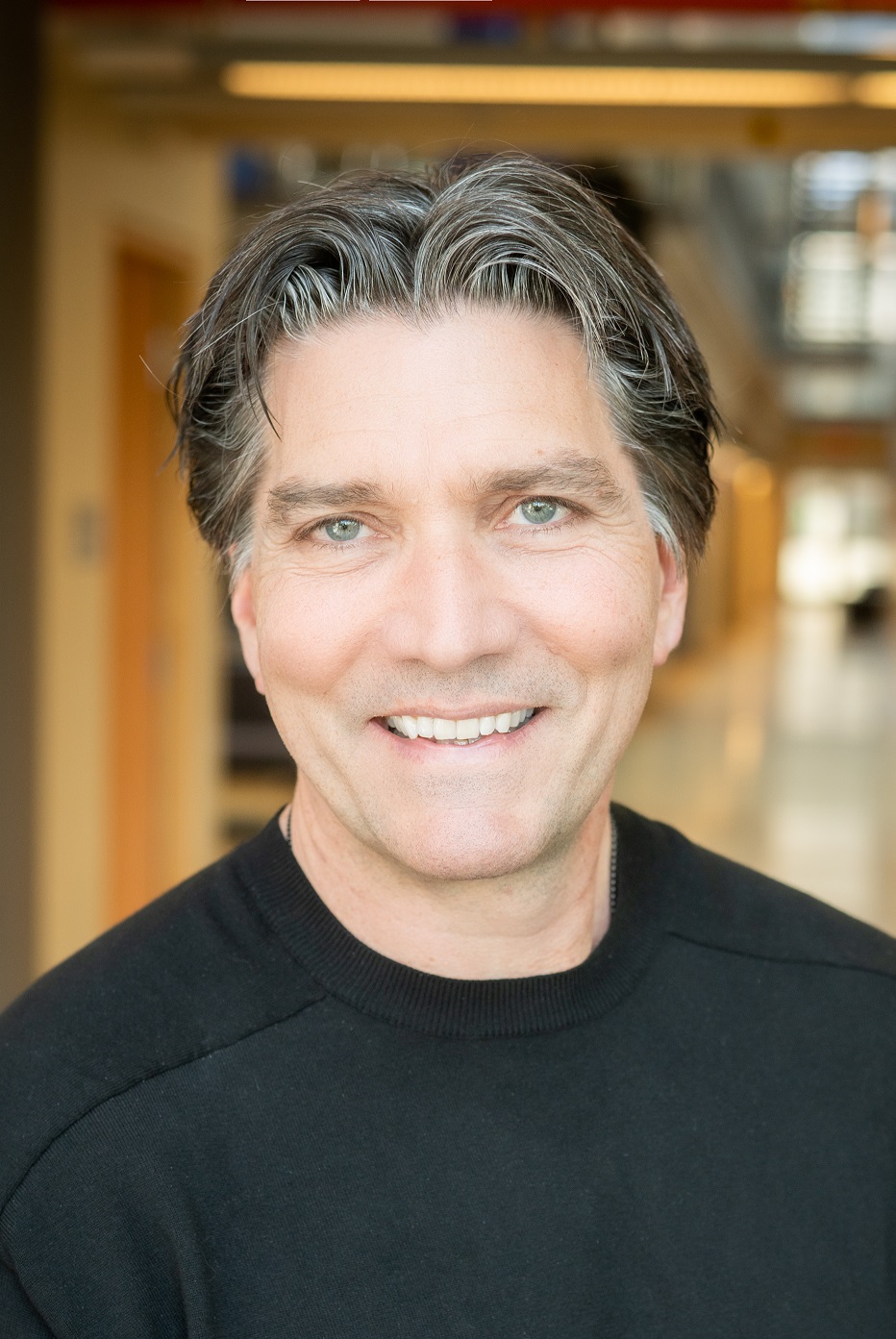}}]{Martin Guay} received the Ph.D. degree from Queen’s University, Kingston, ON, Canada, in 1996. He is currently a Professor in the Department of Chemical Engineering at Queen’s University. His current research interests include nonlinear control systems, especially extremum-seeking control, nonlinear model predictive control, adaptive estimation and control, and geometric control. 

He was a recipient of the Syncrude Innovation Award, the D. G. Fisher from the Canadian Society of Chemical Engineers, and the Premier Research Excellence Award. He is a Senior Editor of IEEE Control Systems Letters. He is the Editor-in-Chief of the Journal of Process Control. He is/was also an Associate Editor for IEEE Transactions on Automatic Control, Automatica, the Canadian Journal of Chemical Engineering, and Nonlinear Analysis: Hybrid Systems.
  \end{IEEEbiography}
  
\end{document}